\documentclass[11pt]{article}

\usepackage{fullpage,wrapfig}
\usepackage{amsmath, amsthm, amssymb, algorithm,thmtools,algpseudocode,enumerate,caption,framed}
\usepackage{paralist, sidecap}
\usepackage[usenames,dvipsnames,svgnames,table]{xcolor}
\definecolor{darkgreen}{rgb}{0.0,0,0.9}
\usepackage[colorlinks=true,pdfpagemode=UseNone,citecolor=Blue,linkcolor=Red,urlcolor=Black,pagebackref]{hyperref}
\usepackage{tikz}
\usepackage[T1]{fontenc}
\usepackage[utf8]{inputenc}
\usepackage{authblk}

\newtheorem{theorem}{Theorem}[section]

\newtheorem{lemma}{Lemma}[section]
\newtheorem{observation}{Observation}[section]

\newcommand{\eh}[1]{h #1}
\newcommand{\vy}[1]{y #1}
\newcommand{\ex}[1]{x #1}

\title{Evacuating Equilateral Triangles and Squares in the Face-to-Face Model\thanks{A preliminary version of this paper appeared in proceedings of the 21st International Conference on Principles of Distributed Systems (OPODIS 2017)~\cite{ChuangpishitMNO17}. Research of Lata Narayanan and Jaroslav Opatrny is supported in part by Natural Sciences and Engineering Research Council of Canada (NSERC).}}

\author[1]{Huda Chuangpishit}
\author[2]{Saeed Mehrabi}
\author[3]{Lata Narayanan}
\author[3]{Jaroslav Opatrny}
\affil[1]{{\small Department of Mathematics, Ryerson University, Toronto,  Canada.

			\texttt{hoda.chuang@gmail.com}}}
\affil[2]{{\small School of Computer Science, Carleton University, Ottawa, Canada.

			\texttt{saeed.mehrabi@carleton.ca}}}
\affil[3]{{\small Department of Computer Science, Concordia University, Montreal, Canada.

\texttt{\{lata,opatrny\}@cs.concordia.ca}}}

\date{}

\begin{document}

\maketitle

\begin{abstract}
Consider $k$ robots initially located at a point inside a region $T$. Each robot can move anywhere in $T$ independently of other robots with maximum speed one. The goal of the robots is to \emph{evacuate} $T$ through an exit at an unknown location on the boundary of $T$.  The objective is to minimize the \emph{evacuation time}, which is defined as the time the {\em last} robot reaches the exit. We consider the {\em face-to-face} communication model for the robots: a robot can communicate with another robot only when they meet in $T$.

In this paper, we give upper and lower bounds for the face-to-face evacuation time by $k$ robots that are initially located at the centroid of a unit-sided equilateral triangle  or square. For the case of a triangle with $k=2$ robots, we give a lower bound of $1+2/\sqrt{3} \approx 2.154$, and an algorithm with upper bound of 2.3367 on the worst-case evacuation time. We  show that for any $k$, any algorithm for evacuating $k\geq 2$ robots  requires at least $\sqrt{3}$ time. This bound is asymptotically optimal, as we show that even a straightforward strategy of evacuation by $k$ robots gives an upper bound of $\sqrt{3} + 3/k$. For $k=3$ and $4$, we give better algorithms with evacuation times of 2.0887 and 1.9816, respectively. For the case of the square and $k=2$, we give an algorithm with evacuation time of $3.4645$ and show that any algorithm requires time at least $3.118$ to evacuate in the worst-case. Moreover, for $k=3$, and $4$, we give algorithms with evacuation times 3.1786 and 2.6646, respectively. The algorithms given for $k=3$ and $4$ for evacuation in the triangle or the square can be easily generalized for larger values of $k$.
\end{abstract}

\section{Introduction}
\label{sec:introduction}
Searching for an object at an unknown location in a specific domain in the plane is a well-studied problem in theoretical computer science~\cite{bookAW,Beck1964,bookBN,bookN,bookS}. The problem was initially studied when there is only one searcher, whom we refer to as a \emph{robot}.
The target is assumed to be a point in the domain, and the robot can only find the target when it visits that point. The goal then is to design a trajectory for the robot that finds the target as soon as possible. Recent work has  focused more on {\em parallel} search by several robots, which can reduce the search time as the robots can distribute the search effort among themselves. The {\em search time} by $k$ robots is generally defined to be the time the {\em first} robot reaches the target. 

A natural generalization of the parallel search problem, called the {\em evacuation problem}, was recently proposed in \cite{CzyzowiczGGKMP14}: consider several robots inside a region that has a single {\em exit} at an unknown location on its boundary. All robots need to reach the exit; i.e., \emph{evacuate} the region, as soon as possible. This is essentially the parallel search problem where the exit is the search target, however we are interested in minimizing the time the {\em last} robot arrives at the exit. Since then, the evacuation problem have been studied in several papers for different regions and types of communication among the robots \cite{BrandtFRW17, BW2017,CzyzowiczGKNOV15,CzyzowiczKKNOS15}.

The time needed for evacuation  substantially depends on the way robots can communicate among themselves. Two models of communication have been proposed: in the {\em wireless model}, each robot can communicate wirelessly with the other robots instantaneously, regardless of their locations. In the {\em face-to-face model}, two robots can communicate with each other only when they meet; i.e., when they occupy the same location at the same time. Since in the wireless model robots can communicate with each other regardless of their locations, as soon as a robot finds the exit, it can announce it to the other robots, which can  then take a straight-line path to the exit. This is not possible in the face-to-face communication model. In this case either a robot that found the exit must intercept other robot(s), or the trajectories of robots need to have a possibility of {\em meeting} so that  information about the location of the exit can be shared. This limited communication capability makes  the design of trajectories of robots of the evacuation problem more challenging.

In this paper, we study the problem of evacuating a unit equilateral triangle and a unit-side square in the face-to-face model with $k\geq 2$ robots,  all of which are initially located at the centroid of the triangle or the square. Our objective is to design the trajectories of the robots so as to minimize the \emph{worst-case evacuation time}, which is defined as the time it takes for \emph{all the robots} to reach the exit.

\paragraph{Related work.} A classical problem related to our paper is the well-known \emph{cow-path} problem introduced by A. Beck \cite{Beck1964}, in which a cow searches for a hole in an infinite linear fence. An optimal deterministic algorithm for this problem and for its generalization to several fences is known, e.g., Baeza-Yates et al.~\cite{Baeza-YatesCR88}. Since then several variants of the problem have been studied~\cite{Baeza-YatesS95, BrandtFRW17, ChrobakGGM15,CzyzowiczKKNOS15, FlocchiniPSW05,GhoshK10,JezL09,LiC09a,Lopez-OrtizS01}. 

Lopez-Ortiz and Sweet~\cite{Lopez-OrtizS01} asked for the worst-case trajectories of a set of robots searching {\em in parallel} for a target point at an unknown location in the plane. Feinerman et al.~\cite{FeinermanKLS12} (see also~\cite{FeinermanK13}) introduced a similar problem in which a set of robots that are located at a cell of an infinite grid and being controlled by a Turing machine (with no space constraints) need to find the target at a hidden location in the grid.  In these two models of multi-robot searching, the robots cannot communicate at all. By controlling each robot by an asynchronous finite state machine, Emek et al.~\cite{EmekLUW14} studied this problem in which the robot can have a ``local'' communication in some sense and proved that the collaboration performance of the robots remains the same, even if they possess a constant-size memory. Lenzen et al.~\cite{LenzenLNR14} extended this problem by introducing the \emph{selection complexity} measure as another factor in addition to studying the time complexity of the problem.

The evacuation problem with several robots has been studied in  recent years under wireless and face-to-face models of communications. For the wireless model,  Czyzowicz et al.~\cite{CzyzowiczGGKMP14} studied the problem of evacuating a unit disk, starting at the center of the disk. They gave a tight bound of $1+2\pi/3 +\sqrt(3)\approx 4.83$ for the evacuation time of two robots, as well as upper and lower bounds of, respectively, 4.22 and 4.159 for three robots. These bounds for $k$ robots become $3+\pi/k+O(k^{-4/3})$ and $3+\pi/k$, respectively~\cite{CzyzowiczGGKMP14}. Czyzowicz et al.~\cite{CzyzowiczKKNOS15} also studied the evacuation problem for $k$ robots for unit-side squares and equilateral triangles in the wireless model. For a unit-side square, they gave optimal algorithms for evacuating $k=2$ robots when located at the boundary of the square. Moreover, for an equilateral triangle, they gave optimal evacuation algorithms for $k=2$ robots in any initial position on the boundary or inside the triangle. They also showed that increasing the number of robots cannot improve the evacuation time when the starting position is on the boundary, but three robots can improve the evacuation time when the starting position is the centroid of the triangle. Recently, Brandt et al.~\cite{BrandtFRW17} considered the evacuation problem for $k$ robots on $m$ concurrent rays under the wireless model. Finally, the evacuation problem on a disk with three robots at most one of which is  \emph{faulty} was recently studied by Czyzowicz et al.~\cite{CzyzowiczGGKKRW17}, the case of several exit on a disk with two robots was considered in~\cite{CzyzowiczDGKM18}, and the priority evacuation of a specific robot in~\cite{CzyzowiczGKKKNO18}, all these  under the wireless model.

For the face-to-face model, Czyzowicz et al.~\cite{CzyzowiczGGKMP14} gave upper and lower bounds of, respectively, 5.74 and 5.199 for the evacuation time of two robots initially located at the center of a unit disk. Both the upper and lower bounds were improved by Czyzowicz et al.~\cite{CzyzowiczGKNOV15} to 5.628 and 5.255, respectively. A further improvement of the upper bound by 0.003 is given in \cite{BW2017}.  Closing the gap between upper and lower bounds remains open. When $k=3$ the upper and lower bounds for the face-to-face model are $5.09$ and $5.514$, respectively, and $3+2\pi/k$ and $3+2\pi/k-O(k^{-2})$ for any $k>3$~\cite{CzyzowiczGGKMP14}.

\paragraph{Our results.} In this paper, we study the evacuation of $k$ robots from an equilateral triangle and unit-side square under the face-to-face model. We prove the following results for an equilateral triangle $T$:
\begin{itemize}
\item  For $k=2$ we prove a lower bound of 2.154 on the evacuation time. We use {\em Equal-Travel with Detour} strategy to get  an evacuation algorithm with evacuation time of 2.3866 in which trajectories include a {\em detour} inside the triangle  before the entire boundary is  explored by the robots. We then show that a further improvement can be obtained  by  algorithms that uses detours {\em recursively}. The algorithm with two detours  improves the evacuation time for two robots to 2.3367.
\item For $k\geq 3$, we show that any algorithm for evacuating $k$ robots from triangle $T$ requires at least $\sqrt{3}$ time. We prove that this bound is asymptotically optimal by giving a simple algorithm that achieves an upper bound of $\sqrt{3} + 3/k$. 
\item We show that a significant improvement on the above upper bound can be obtained using the {\em Equal-Travel  Early-Meeting} strategy. In this strategy the trajectories  of all robots are of the same length and they include a \emph{meeting point} inside the triangle for all robots before the entire boundary is explored to share information about the exploration results so far. Applying this strategy we design algorithms for $k=3$, $4$, and $5$ with evacuation times of $\approx$  $2.0887$, $1.982$, and $1.876$, respectively.
\end{itemize}

We then study the face-to-face evacuation problem in a unit-side square $S$:
\begin{itemize}
\item For $k=2$, we first prove a lower bound of $3.118$ on the evacuation time of any deterministic algorithm. Then, we give an  {\em Equal-Travel with Detour} algorithm that achieves the evacuation time of 3.46443. The possibility of recursive application of detours is pointed out. 
\item For $k=3,4$, we give  {\em Equal-Travel  Early-Meeting} algorithms with evacuation times of 3.178 and 2.664, respectively.
\end{itemize}

\paragraph{Organization.} We first specify some preliminaries and notation in Section~\ref{sec:prelimins}. Then, we present our results for the equilateral triangle; we give the proofs of our lower bounds in Section~\ref{sec:lowerBounds}, and our evacuation algorithms in Section~\ref{sec:upperBounds}.  Section~\ref{sec:square} contains our lower and upper bounds results for the unit-side square. Finally, we conclude the paper with a summary of our results and a discussion of open problems in Section~\ref{sec:conclusion}.

\section{Preliminaries}
\label{sec:prelimins}
For two points $p$ and $q$ in the plane, we denote the line segment connecting $p$ and $q$ by $pq$ and its length  by $|pq|$. We assume the robots are initially all located at a point in the given region, the exit is located at an unknown location on the boundary of the region, and the robots communicate using the face-to-face model. Every robot moves at maximum speed $1$.

A deterministic algorithm for the evacuation problem by $k$ robots takes as input a specification of the given region (a triangle or a square in our case)  and the point $O$ where the $k$ robots located initially. It outputs for each robot a {\em fixed trajectory} consisting of a sequence of connected line segments or curves to be followed. We assume every robot knows the trajectories of all the robots. A robot $R$ follows its trajectory unless:
\begin{itemize}
\item $R$ sees the exit: $R$ may then quit its trajectory and go to a point where it can intercept another robot (or other robots)  and inform it (or them) about the exit. 
\item $R$ meets another robot who has found the exit: $R$ then quits its trajectory and proceeds directly to the exit.
\end{itemize}

Observe that the robots are initially co-located, and the initial part of their trajectories may be identical; i.e., when going to the boundary of the triangle or the square. Later on, the trajectories of two or several robots may intersect and the intersection point may be reached by all robots at the same time. We call such a point a {\em meeting point}.
A meeting point might be in the interior of the triangle or the square, and it can serve as a place for the robots to exchange information about the progress in the search for exit until this time. The meeting point can be reached by robots before the whole boundary is explored. When the robots meet, if one of the robots has found the exit, they can all proceed towards it. Otherwise, the robots can continue in the search for the exit separately or together. As shown in~\cite{CzyzowiczKKNOS15}, and  in Section~\ref{sec:upperBounds}, an algorithm with a meeting point in the interior of the region can improve the evacuation time in some cases.

If we have two robots; i.e., $k=2$, then the trajectories do not have to contain a meeting point in order to improve the evacuation time. The trajectory of each robot can leave the boundary of the region and return to the boundary without intersecting the other trajectory. We say that the robot makes a {\em detour}.  The detour part of the trajectory is designed so that the robot can be intercepted during the detour by the other robot in case the other robot has found the exit already. 
In the absence of such an interception, the robot has gained information about the {\em absence}  of the exit in some part of the boundary. Since the trajectories do not need to intersect, it allows the trajectories to be shorter. The trajectories of two robots  with one detour was used in \cite{CzyzowiczGGKMP14}, and also in \cite{BW2017} to improve the evacuation time of two robots in the circle. We use this strategy in Subsection \ref{ssec:2robots} as well as in Section~\ref{sec:square}, and show it can be used recursively for a further improvement.

\section{Evacuation of the Equilateral Triangle}
Throughout this section we denote an equilateral triangle by $T$ and its vertices  by $A, B$, and $C$. Thus we sometimes write $ABC$ to refer to $T$. We always assume that the sides of $T$ have length $1$. Throughout the paper we use the following triangle terminology, see Figure~\ref{fig:terminology}:
\begin{itemize}
\item  By $h$ we denote the {\em height} of the equilateral triangle. Observe that $\eh=\sqrt{3}/2$.
\item   We denote by $O$ the {\em centroid} of $T$ (i.e., the intersection point of the three heights of $T$).
\item We use $\ex$, and $\vy$ to denote the distance of $O$ to a vertex, and to the side of the triangle,  respectively; notice that $\ex=2\eh/3=\sqrt{3}/3$ and  $\vy=\eh/3=x/2=\sqrt{3}/6\approx 0.288$. 
\end{itemize}

\begin{figure}[t]
\centering
\includegraphics[width=0.35\textwidth]{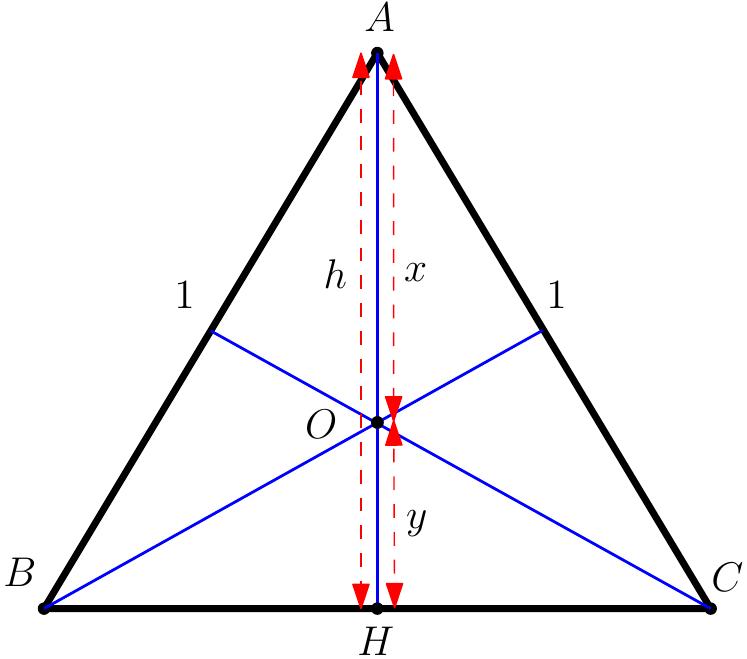}
\caption{Our triangle notation.}
\label{fig:terminology}
\end{figure}

We define $E_{\mathcal  {A}}(T, k)$  to be the worst-case evacuation time of the unit-sided equilateral triangle $T$  by $k$ robots using algorithm $\mathcal{A}$, and we define $E^*(T, k)$  to be the {\em optimal} evacuation time of the triangle  by $k$ robots in the face-to-face model.

\subsection{Lower Bounds}
\label{sec:lowerBounds}
In this subsection, we  prove lower bounds on the evacuation time of the equilateral triangle. We first show that regardless of the number of robots, $\sqrt{3}$ is a lower bound on the evacuation time.  This bound holds even if the exit is {\em a priori}  known to be at one of the three vertices of the triangle. 
\begin{theorem}
\label{th:krobots}
Consider $k$ robots $R_1, R_2,\ldots, R_k$,  initially located at the centroid of an equilateral triangle $T$. In the face-to-face model, the evacuation time of $k$ robots $E^*(T,k) \geq \sqrt{3}\approx 1.732 $.
\end{theorem}
\begin{proof}
Consider an arbitrary evacuation algorithm ${\mathcal A}$ for  $k$ robots. We first run the algorithm to see which vertex is the last one visited by the robots (two or even three vertices could be visited at the same time, as the last ones in which case we choose an arbitrary one as last). Assume without loss of generality that $A$ is the last vertex visited by any of the robots; let $I_1$ be the input in which the exit is at $A$. Consider the execution of the algorithm on input $I_1$, and let $t$ be the time the second of the three vertices is visited by some robot $R$. Without loss of generality, let this second vertex be $B$; that is, $R$ visits vertex $B$ at time $t$ on input $I_1$. Let $I_2$ be the input where the exit is at the remaining vertex $C$. We argue that the evacuation time of the algorithm must be $\geq 3x$ on one of these two inputs, where $x$ is the distance of the centroid to a vertex.

If $t\geq 3x-1$, then it takes  additional time 1 for robot $R$ to reach the exit at $A$, leading to a total evacuation time of at least $3x$ on input $I_1$. Therefore, assume that $t<3x-1$. We claim that since $R$ has to reach $B$ before time $3x-1$, it is impossible for $R$ to meet a robot $R'$ that has already visited $A$ or $C$ before $R$ reaches $B$ at time $t$. Suppose $R$ was able to meet $R'$ that had visited $A$ (without loss of generality) at some meeting point  $M$ at time $t_M$. Then clearly $t_M \geq x + |AM|$. After meeting $R'$, the robot $R$ needs time at least $|MB|$ to get to $B$. We conclude that $t \geq t_M + |MB| \geq x + |MB| + |MA| \geq x + 1$. However, $x+1 > 3x-1$, a contradiction. Thus, $R$'s trajectory to $B$, reaching $B$ at time $t < 3x-1$ cannot allow a meeting between $R$ and any robot
that has already visited $A$ or $C$. Therefore, the behaviour of the robot $R$ must be the same on inputs $I_1$ and $I_2$ until time $t$ when $R$ arrives at $B$. Observe now that $t \geq x$. At time $2x$, if the robot $R$ is at distance $>x$ from $A$, the adversary puts the exit at $A$ (input $I_1$), and if it is at distance $> x$ from $C$, it puts the exit at $C$. Combined with the fact that at time $2x$, the robot $R$ can travel at most distance $2x-t\leq x$ from $B$, we have the desired result.
\end{proof}

The above bound is asymptotically optimal, as we will describe in Section~\ref{sec:upperBounds} a simple algorithm  that evacuates $k$ robots in $\sqrt{3} + 3/k$ time from an exit situated anywhere on the boundary. We remark also that  in the wireless communication model, $E^*(T,6) = \frac{2\sqrt{3}}{3}$~ (D.~Krizanc, private communication, 2015), showing that for the equilateral triangle, evacuation even with arbitrarily many robots takes much more time in the face-to-face model, than evacuating only six robots in the  wireless model.

When $k=2$, we are able to prove a better lower bound of $1+2\sqrt{3} \approx 2.15$. The argument used for the lower bound is an adversary argument: depending on what the algorithm does, the adversary places the exit in such a way so as to force at least the claimed evacuation time. The key insight can be summarized as follows: if an algorithm is to do better than the claimed lower bound, either the robots cannot meet in a useful way to shorten the time to reach the exit, or they simply cannot finish the exploration. To this end, we first prove  the following technical lemma. We first need some notation. For the equilateral triangle $ABC$, let $D, E$ and $F$ denote the midpoints of sides $AB, AC$ and $BC$, respectively, and let $\mathcal{S}=\{A,B,C,D,E,F\}$. We say two points in $\mathcal{S}$ have {\em opposite positions} if one point is a vertex of the triangle $T$ and the other point is located on the opposite side of that vertex. For example, the vertex $C$ and a point in $\{A, D, B\}$ have opposite positions.
\begin{lemma}[Meeting Lemma]
\label{lem:noMeetingBeforeT}
Consider a deterministic algorithm $\mathcal{A}$ for evacuating two robots in an equilateral triangle $T$, and let $p_1, p_2\in\mathcal{S}$ have opposite positions. If $\mathcal A$ specifies a trajectory for one of the robots in which it visits $p_1$ at time $t'$ and a trajectory for the other robot in which it visits $p_2$ at time $t$ such that
$|t-t'| <h$, then the two robots cannot meet between time $t$ and $t'$. \end{lemma}
\begin{proof}
We may assume that $t' < t$. Suppose for a contradiction that the robots meet at time $t_m$  at some point $z$, where  $t'< t_m\leq t$. Since $p_1$ and $p_2$ have opposite positions $|p_1p_2|\geq \eh$. Therefore, $|p_1z|+|zp_2|\geq \eh$. Moreover $|p_1z|\leq t_m-t'$ and $|zp_2|\leq t-t_m$. This implies that
\[
\eh\leq |p_1z|+|zp_2|\leq (t_m-t')+(t-t_m)=t-t'< \eh
\]
which is a contradiction.
\end{proof}

\begin{figure}[t]
\centering
\includegraphics[width=0.85\textwidth]{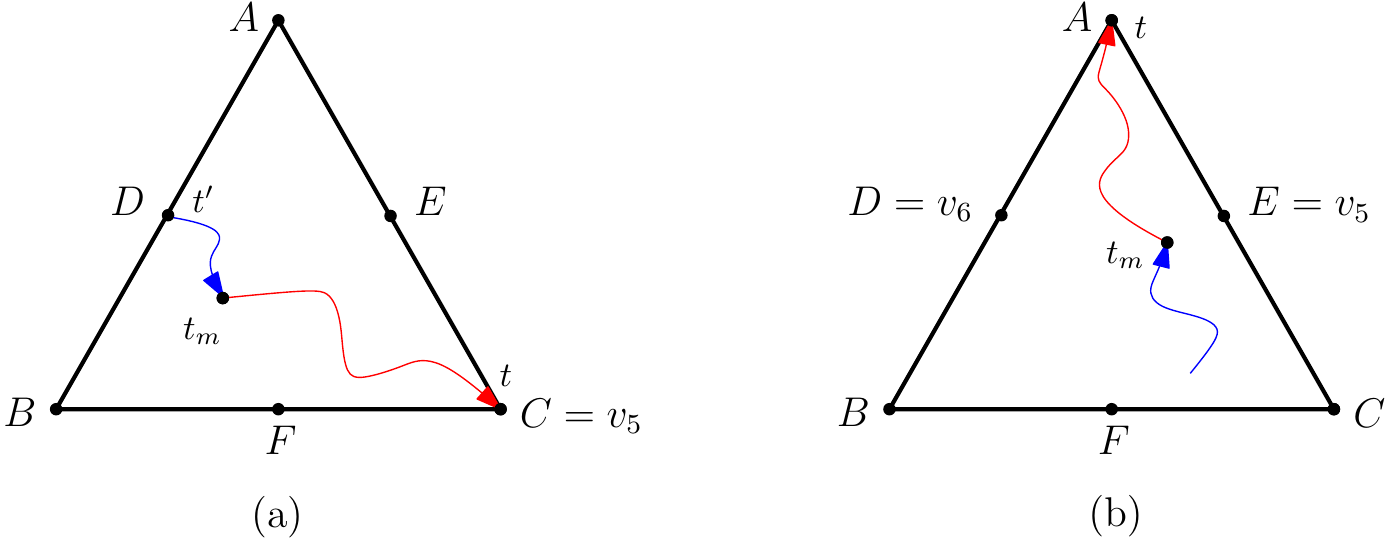}
\caption{(a) An illustration in support of the proof of the Meeting Lemma. (b) An illustration in support of case 2 in the proof of Theorem~\ref{lem:twoRobotsLowerBound}.}
\label{fig:meetingLem}
\end{figure}

\begin{theorem}
\label{lem:twoRobotsLowerBound}
Consider $2$ robots $R_1, R_2$,  initially located at the centroid of an equilateral triangle $T$. If the robots communicate using the face-to-face model, then the evacuation time of the two robots $E^*(T,2) \geq 1+4y = 1+2/\sqrt{3}\approx 2.154 $.
\end{theorem}
\begin{proof}
Suppose for the purpose of contradiction,  that there is an algorithm $\mathcal A$ for evacuation by two robots, such that $E_{\mathcal A}(T, 2) < 1 + 4y$. We first focus attention on the set of points ${\cal S} = \{A, B, C, D, E, F \}$. There exists some input $I$ on which the exit is the last point in $\cal{S}$ to be visited by either of the robots, according to the trajectories specified by $\mathcal{A}$. Let $t_5$ be the time the {\em fifth} point of $\cal{S}$ is visited by a robot on input $I$.  Let $v_1,v_2,\dots,v_6$ be the order in which the points in $M$ are visited by the robots, on input $I$; the exit is at $v_6$. Without loss of generality assume that $v_5$ is visited by robot $R_1$. Thus, $v_6$ is not yet visited before  time $t_5$; it may be visited at or after time $t_5$. First, note that since at least five  points are visited at or before time $t_5$, one of the robots must have visited at least three points in $M$. It follows that $t_5\geq 1+\vy$. If $t_5\geq 0.5+4\vy$, since the exit is at $v_6$, which is at least 0.5 away from $R_1$, we obtain $E_{\mathcal A}(T,2) \geq 1+ 4\vy$, a contradiction. We conclude that $1+\vy\leq t_5<0.5+4y$.

We now consider the following exhaustive cases depending on whether $v_5$ is a vertex of $T$ or a midpoint of a side of $T$.

\noindent\emph{Case 1. $v_5$ is a vertex of $T$.} Without loss of generality assume that $v_5$ is $C$. See Figure~\ref{fig:meetingLem}(a). If $v_6$ is any of $A, D, B$, then at time $t_5$, robot $R_1$ needs time at least $\eh$ to arrive to $v_6$, which implies that  $E^*(T,2)\geq t_5+\eh\geq 1+4\vy$, a contradiction. So we conclude that  $v_6$ is at either $E$ or $F$. Since $t_5<0.5+4\vy$, robot $R_1$ could have visited at most one of $A, D, B$ by time $t_5$.
This means that $R_2$ must have visited at least two of $A, D, B$.   Let $v$ be the second vertex of the set $A, D, B$ to be visited by $R_2$, and assume it arrives there at time $t'$.  Note that $0.5+4y>t'\geq 0.5+\vy$. By the Meeting Lemma, the two robots do not meet at any time between $t'$ and $t_5$ on input $I$. 

Now consider an input $I'$ in which the exit is at $v$. Clearly the robots behave identically on both inputs $I$ and $I'$ until time $t'$. After this time, $R_2$ on seeing the exit at $v$ may behave differently; however robot $R_1$ must behave exactly as in $I$ unless it meets robot $R_2$, which by the Meeting Lemma, cannot happen before time $t_5$. Therefore, after time $t_5$, it takes at least an additional $h$ to reach the exit at $v$, giving a total evacuation time of at least $t_5 + h \geq 1 +4y$, a contradiction.

\noindent\emph{Case 2. $v_5$ is a midpoint of a side of $T$, and $v_6$ is another midpoint: } Without loss of generality assume that $v_5$ is $E$, and $v_6$ is $D$; see Figure~\ref{fig:meetingLem}(b). Then, all three vertices must have been visited before or at time $t_5$. Since $R_1$ cannot visit two vertices {\em before} arriving at $E$ at time $t_5 < 0.5+ h + y=0.5+4y $, we conclude that $R_2$ must  visit two vertices by time $t_5$. Referring to Figure~\ref{fig:meetingLem}(b), consider the second vertex visited by $R_2$. Observe that $R_2$ cannot arrive there before time $1+x$. (i) If it is $B$, then we put the exit at $E$. This way, $R_2$ needs time at least $\eh$ to get to $E$ from $B$ and so $E^*(T,2)\geq 1 + x + \eh= 1+ 5\vy$. (ii) If it is $C$, then we put the exit at $D$. Then, $R_2$ needs time at least $\eh$ to get to $D$ from $C$ and so $E^*(T,2)\geq 1+\ex+\eh= 1+ 5\vy$.

We conclude that the second vertex visited by $R_2$ must be $A$. We first note that $R_2$ cannot have visited all of $F, B, A$ or all of $F, C, A$ by time $t_5$, as visiting either set of three points takes time at least $0.5 + x + h > 0.5 + y + h > t_5$. So, $R_1$ must have visited $F$ and $C$ (or $F$ and $B$) before coming to $E$. Let $P$ be the second of the two points visited by $R_1$ before going to $E$. Note that $P$ could be $F$, $C$, or $B$, and suppose $R_1$ visits $P$ at time $t'$.  Clearly, $t'\geq 0.5+\vy$. Since $R_2$ must visit $A$ at or before time $t_5 < 0.5 + 4y$, by the Meeting Lemma, the robots cannot meet at any time  $t'\leq t_m\leq t_5$.  In other words, $R_1$ cannot meet $R_2$ after the latter has visited $P$ and reach $A$ on time. Now consider the input $I'$ in which the  exit is at $P$.
On input $I'$, the robot $R_2$ will have the same behavior as on input $I$ until it reaches $A$ at time $t>1+x$ and then needs to get to the exit (which is at  $F$ or $C$ or $B$). Therefore we have $E^*(T,2)\geq 1+x+\eh= 1+ 5\vy$, a contradiction.

\noindent\emph{Case 3. $v_5$ is a midpoint of a side of $T$, and $v_6$ is a vertex: }  Without loss of generality assume that $v_5$ is $E$. In what follows, we use $t_p$ to denote  the time a point $p$ was visited for the first time by a robot. Then clearly $t_E =t_5$. If $v_6$ is $B$, then $R_1$ needs time $\geq h$ to reach the exit, so on input $I$, the evacuation time is at least $t_5+h \geq 1+h+y$. Therefore, assume without loss of generality that $v_6$ is $A$; see Figure~\ref{fig:meetingLem}(b).
If a single robot visits both $B$ and $C$, it takes time at least $2+x > 1+4y$ to reach vertex $A$. Therefore, $B$ and $C$ must be visited by different  robots. We consider separately the two cases: $R_1$ visits $C$ and $R_2$ visits $B$;  and $R_1$ visits $B$ and $R_2$ visits $C$. 
 
Suppose $R_1$ visits $C$ before visiting $E$, and $R_2$ visits $B$. First observe that $R_1$ cannot also visit $D$, as visiting $C$, $D$, and $E$ takes time at least $0.5+4y$, a contradiction to $t_E < 0.5+4y$. Therefore $R_2$ must visit $D$ in addition to $B$. Either $R_1$ or $R_2$ must visit $F$. If $R_2$ visits $F$, Lemma~\ref{R1-CE} assures that $E_{\mathcal{A}}(T,2) \geq 1+4y$ and if $R_2$ does not visit $F$, Lemma~\ref{R1-CFE} does the same.  

Suppose instead that $R_1$ visits $B$ before visiting $E$ and $R_2$ visits $C$. Then $R_1$ cannot visit both $D$ and $F$, as visiting $D, B, F, E$ takes time at least $1.5+y > 0.5 + 4y$, Lemmas~\ref{R1-BFE}, ~\ref{R1-BDE}, and ~\ref{R1-BE} now assure that $E_{\mathcal{A}}(T,2) \geq 1+4y$ for the cases when $R_1$, in addition to visiting $B$ and $E$,  visits $F$, visits $D$, and visits neither, respectively. 
\end{proof}

For all the lemmas below, we assume that according to algorithm $\mathcal{A}$, we have $v_5 = E$ and $v_6 = A$, and robot $R_1$ visits $E$ at time $1+y \leq t_E < 0.5+4y$. We start with a simple observation that is used repeatedly.

\begin{observation}
\label{two-points}
Let $p$ be a point on the boundary. If at time $1+4y - |Ap|$, both $A$ and $p$ are unvisited then $E_{\mathcal{A}}(T,2) \geq 1+4y$.
\end{observation}
\begin{proof}
Put the exit at whichever of the two points is visited later. Since at time $1+4y - |Ap|$, neither is visited, the time to evacuate is at least $1+4y-|Ap|+|Ap| = 1+4y$.
\end{proof}

\begin{lemma}
 If $R_2$ visits $B, D$ and $F$, and $R_1$ visits $C$ and $E$, then  $E_{\mathcal{A}}(T,2) \geq 1 + 4y$. 
\label{R1-CE}
\end{lemma}
\begin{proof}
First, observe that if $B$ is not visited {\em first } of the three points $B, D, F$, then $t_B \geq 0.5+y$. Then, since $E$ is visited by $R_1$ and $0.5 +4y >t_E \geq 1+y$, by the Meeting Lemma, $R_1$ and $R_2$ cannot meet between $t_B$ and $t_E$. Thus if the exit is at $B$, it will take $R_1$ time at least $t_E + h \geq 1+4y$ to reach there. We conclude that $B$ must be visited first. If $R_2$ visits $B$, $D$, and $F$ in that order, then $t_F \geq 1+x$, so by Observation~\ref{two-points}, we have $E_{\mathcal A}(T,2) \geq 1+4y$. So, $R_2$ must visit $B, F$ and $D$ in this order. 

Let $P$ be the closest point from $B$ on the $BD$ line segment that is not visited by $R_2$ before it visits $F$. Then the time for $R_2$ to reach $F$ is at least $|OP| + |PB| + |BF|$ (Figure \ref{fig:R1-CE}, the blue trajectory). Therefore, the earliest time $R_2$ can reach $P$ is $|OP| + |PB| + |BF| + |FP|$. It can be verified that for any point $M$ on the $BD$ line segment, this time is more than $1+4y-|AM|$, therefore it is true for the point $P$ defined above.
Also, $R_1$ cannot visit $P$ on time: if it visits $P$ before $C$ (Figure \ref{fig:R1-CE}, the green trajectory), we have $t_E \geq |OP| + |PC| + |CE| \geq |OD| + |DC| + |CE| = 0.5+4y$, and if it visits $C$ before $P$ (Figure \ref{fig:R1-CE}, the red trajectory), we have $t_E \geq |OC| + |CP| + |PE| \geq |OC| + |CD| + |DE| = 2y + 3y + 0.5$. Thus neither robot can visit $P$ before time $1+4y - |AP|$. The lemma now follows from Observation~\ref{two-points}.
\end{proof}

\begin{figure}[t]
\centering
\includegraphics[width=0.35\textwidth]{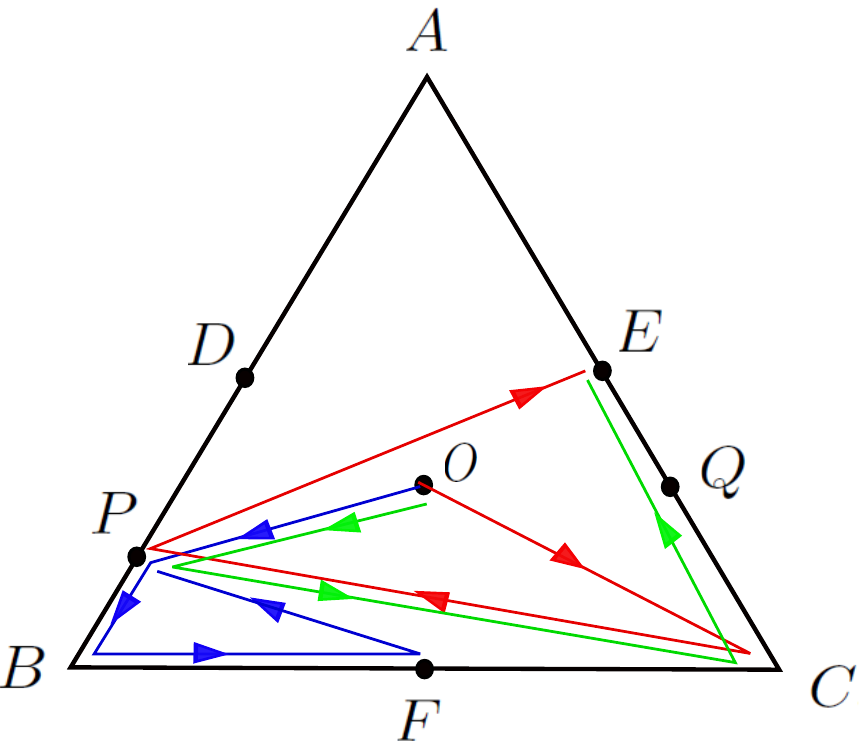}
\caption{An illustration of possible trajectories of $R_1$ and $R_2$, in support of Lemma~\ref{R1-CE}.}
\label{fig:R1-CE}
\end{figure}

\begin{lemma}
 If $R_2$ visits $B$ and $D$ and  $R_1$ visits $C, F,$ and $E$, then  $E_{\mathcal{A}}(T,2) \geq 1 + 4y$. 
\label{R1-CFE}
\end{lemma}
\begin{proof} First observe that in this case $t_E \geq1+y$. If $R_1$ and $R_2$ do not meet between $t_B$ and $t_E$, and if the exit is at $B$, then $R_1$ needs time at least $t_E + h \geq 1+4y$ to reach the exit at $B$, and  $E_{\mathcal{A}}(T,2) \geq 1+4y$. We conclude that $R_1$ and $R_2$ must meet between $t_B$ and $t_E$.
However, if $R_2$ visits $D$ before $B$, then $t_B \geq 0.5+y$. Since $t_E <0.5+4y$, by the Meeting Lemma, $R_1$ and $R_2$ cannot meet between $t_B$ and $t_E$. Therefore, $R_2$ must visit $B$ before $D$. 

Now consider $R_1$'s trajectory, and  assume that $R_1$  visits $C, F$, and $E$ in that order. Using a similar argument as in Lemma~\ref{R1-CE}, we can see that there exists an unvisited point $P'$ on the $CE$ segment at time $1+4y - |AP'|$. It follows from Observation~\ref{two-points} that $E_{\mathcal{A}}(T,2) \geq 1+4y$. 

We conclude that  $R_1$  visits $F$ before $C$.  Notice that $t_C \geq 0.5+y$, and $t_D \leq t_E < 0.5+4y$, so by the Meeting Lemma, $R_1$ and $R_2$ cannot meet between $t_C$ and $t_D$. Consider now the possibility of a meeting after $t_C$ and after $t_D$, as in Figure~\ref{fig:l6}. Clearly, at this time all of segment $BC$ must have been visited; if a point $p$ on $BC$ is unvisited at time $t_D$, then neither $A$ nor $p$ can be visited before time $t_D + \sqrt{3}/4$, which by Observation 4 gives a contradiction. Thus there must exist a point $p$ on $BC$ such that robot $R_1$ explored the segment $pC$ and $R_2$ explored the segment $Bp$. Let $x=|Bp|$. 
Thus $t_B \geq \sqrt{(0.5-x)^2+y^2}+x$, and  $t_C \geq \sqrt{(0.5-x)^2+y^2}+1-x$.

Suppose the exit is at $C$. If $R_2$ proceeds to $A$ without meeting $R_1$, then clearly, $R_2$ cannot reach the exit at $C$ before time $t_b + 2 > 1+4y$. So $R_1$ must be able to meet $R_2$ before the latter can reach $A$. Then note that $R_1$ cannot intercept $R_2$ {\em on} the line segment $AB$, since the time to go there and return to $C$ is at least  $0.5+y+6y\geq 2.5$. So $R_1$ must meet $R_2$  at some point $M$ in the interior of the triangle. Thus, at some point  on the segment $AB$ robot $R_2$ must leave the segment $AB$ to travel to $M$; let $q$ be the last point on the $AB$ segment visited by $R_1$ before it arrives at $M$ at time $t_M$. Further, let $a = |qM|$ and let $b=|MC|$ as shown in Figure~\ref{fig:l6}. 

Now, notice that:
\begin{itemize}
\item  For $R_1$ and $R_2$ to get back to the exit at $C$ in time, it must be that $t_M + b \leq 1+4y$ and since $t_M \geq t_c+b$, we obtain:
$$b< (1+ 4y - t_c)/2 \leq (1+4y- (\sqrt{(0.5-x)^2+y^2}+1-x))/2$$

\item Consider now a point $q'$ infinitesimally close to $q$ on the $qA$ segment. Suppose it is visited at time $t_{q'}$. Then if $t_{q'}  \geq 1 + 4y - |Aq'|$, by Observation 4, we obtain a contradiction. Therefore $t_{q'}< 1 + 4y - |Aq'|$. However, we also have $t_{q'} \geq  t_M + a \geq t_B + |Bq'| + 2a$, which implies that
$$a <  (1 + 4y - |Aq'| - |Bq'| - t_B)/2 \leq (1+4y- (\sqrt{(0.5-x)^2+y^2}+x+1))/2$$
\end{itemize}

This yields 
\begin{align*}
a+b & <  (1+4y-(\sqrt{(0.5-x)^2+y^2}+x+1))/2+(1+4y- (\sqrt{(0.5-x)^2+y^2}+1-x))/2 \\
& =  4y -\sqrt{(0.5-x)+y^2} \leq 3y,
\end{align*}
a contradiction, since $a+b$ is at least the height of the triangle. We conclude that there can be no such meeting point $M$, and therefore, evacuation cannot take place in the claimed time.
\end{proof}

\begin{figure}[t]
\centering
\includegraphics[width=0.4\textwidth]{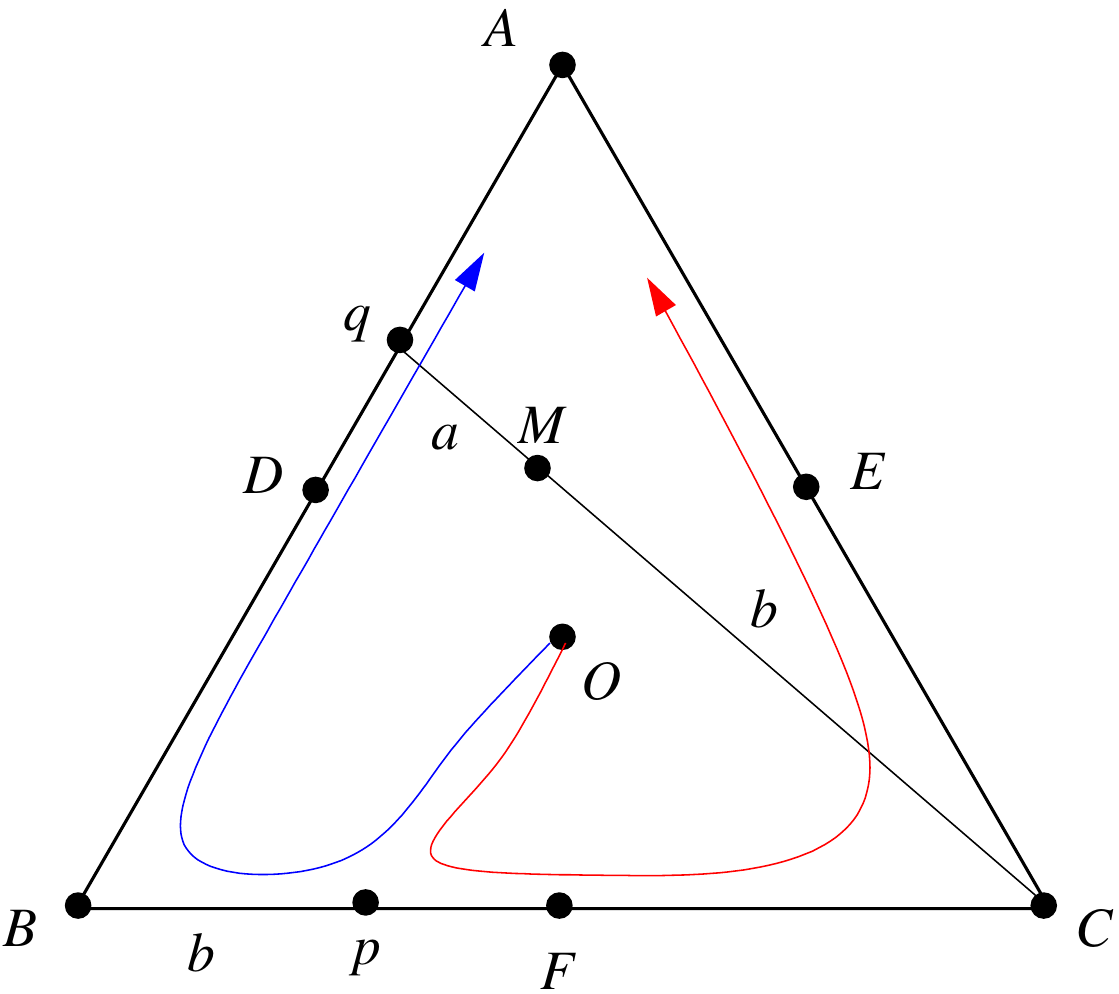}
\caption{Trajectories when  $R_2$ must visit $B$ before $D$ and $R_1$  visits $F$ before $C$.}
\label{fig:l6}
\end{figure}

\begin{lemma}
If robot $R_1$ visits $B, F$ and $E$, and $R_2$ visits $C$ and $D$, then $E_{\mathcal{A}}(T,2) \geq 1+4y$. 
\label{R1-BFE}
\end{lemma}
\begin{proof}
First  observe that $E$ must be visited last, and if $R_1$ visits $F$ before $B$ then $t_E \geq 0.5+4y$. So $B$ must be visited before $F$. Now let $P$ be a point at distance $0.3$ from $B$ on the $BD$ segment, and $Q$ a point at distance $0.34$ from $C$ on the $CE$ segment. It can be verified that if $R_1$ visits a point on the $PD$ segment before arriving at $B$, then $t_E \geq 0.5+4y$. Similarly, $R_1$ cannot visit any point in the $QC$ line segment if it is to reach $E$ by time $0.5+4y$. See Figure \ref{fig:R1-BFE} (a). Therefore the entire $PD$ line segment and the entire $QC$ line segments must be visited by $R_2$. Now we consider the order of visiting $D, Q, C, P$. If $D$ or $P$ are visited before $C$, then $C$ cannot be reached before $4y$ which means that if the exit is at $A$, $R_2$ cannot reach it before time $1+4y$. So, $C$ has to be visited before $P$ or $D$, Figure \ref{fig:R1-BFE} (b). Regardless of whether $Q$ or $C$ is visited first, it can be verified that it is impossible for $R_2$ to reach $P$ before time $1+4y-|AP|$, yielding the desired conclusion, using Observation~\ref{two-points}.
\end{proof}

\begin{figure}[t]
\centering
\includegraphics[width=0.75\textwidth]{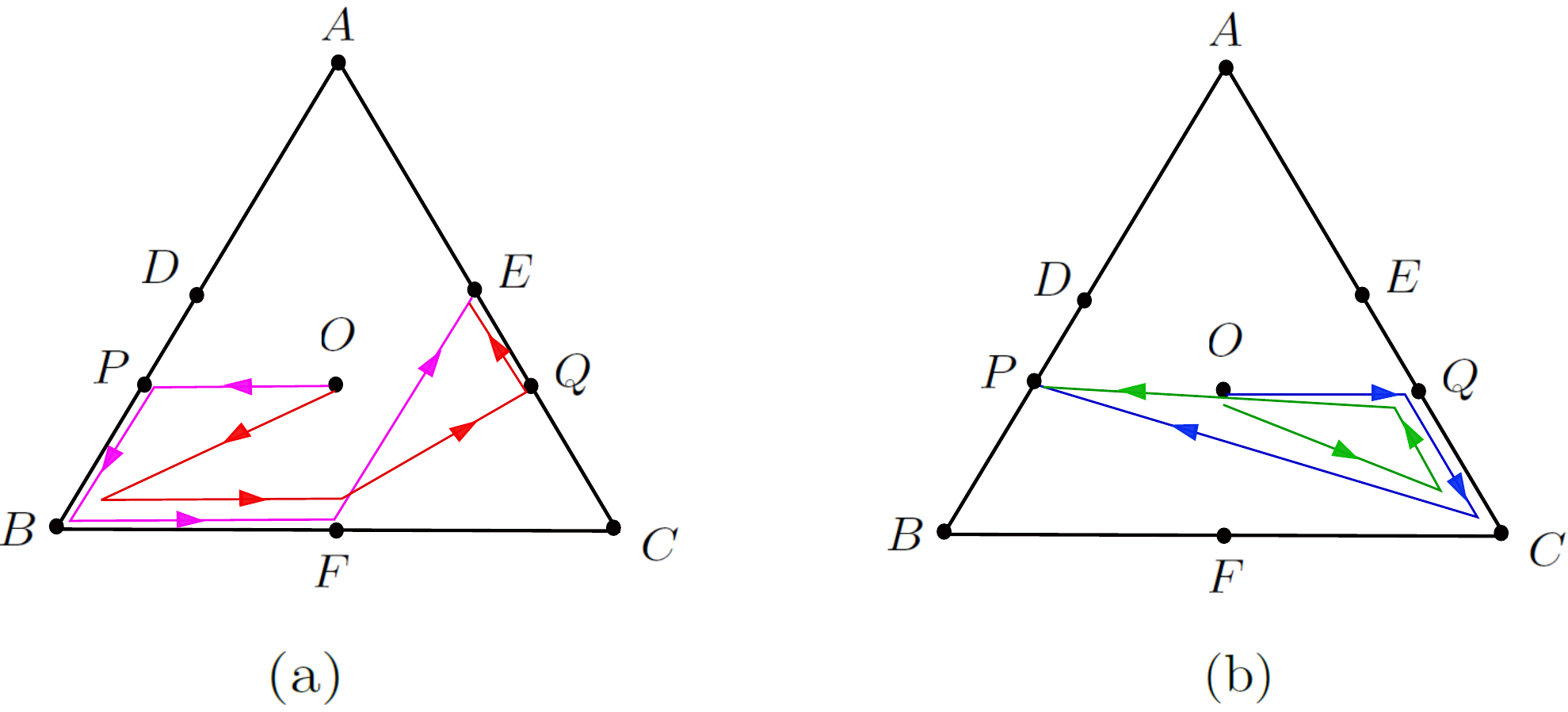}
\caption{(a) An illustration of possible trajectories of $R_1$ and  (b) possible trajectories of $R_2$, in support of Lemma~\ref{R1-BFE}.}
\label{fig:R1-BFE}
\end{figure}

\begin{lemma}
If robot $R_1$ visits $B, D$ and $E$, and $R_2$ visits $C$ and $F$, then $E_{\mathcal{A}}(T,2) \geq 1+4y$.
\label{R1-BDE}
\end{lemma}
\begin{proof}
$R_2$ must visit $F$ before $C$, as otherwise as shown in the proof of Lemma~\ref{R1-CE}, there will be a point $P$ on the $CE$ segment that cannot be visited before time $1+4y-|AP|$. If $R_1$ visits $D$ before $B$, then $t_E \geq 0.5+4y$. So $R_1$ must visit $B$ before $D$. By Observation~\ref{two-points},  the entire $BC$ edge must be visited at or before time $1+y$. Let $Q$ be the leftmost point on the $BC$ edge that is not visited by robot $R_1$. Then $R_2$ must visit the entire $QC$ segment. Since $R_2$ must visit $C$ before time $4y$, we see that $|BQ| > 0.236$. As a result $t_D \geq 1.13155$. Let $R$ be the point at distance 0.05 from $D$ on the the $DA$ segment, and $S$ be the point at distance $0.03$ from  $E$ on the $EC$ segment. Using the assumption that it reaches $E$ before time $0.5+4y$, it can be verified that $R_1$ cannot visit either $R$ or $S$ before time $1.657 > 0.5+4y$, and  $1.661 > 0.5+4y$ respectively.  Then since the $SE$ segment must be visited by $R_2$ before time $0.5+4y$, $R_2$ cannot reach $R$ before time $|OF|+|FC|+|CS| +|SR| > 1.75 > 1+4y - |AR|$. It follows from Observation~\ref{two-points} that $E_{\mathcal{A}}(T,2) \geq 1+4y$.
\end{proof}

\begin{lemma}
\label{R1-BE}
If robot $R_1$ visits $B$ and $E$, and $R_2$ visits $C, F, $ and $D$, then $E_{\mathcal{A}}(T,2) \geq 1+4y$. 
\end{lemma}
\begin{proof}
We observe that $D$ must be visited last; if $F$ is visited last, $t_F \geq 0.5+4y$, and if $C$ is visited last, then $t_C \geq 1+y >4y$. In both cases, Observation ~\ref{two-points} gives the desired result. The rest of the proof is analogous to the case when robot $R_1$ visits $B, F, E$. 
\end{proof}

\subsection{Evacuation Algorithms}
\label{sec:upperBounds}
In this subsection, we give evacuation algorithms for $k$ robots, $k= 2$, $3$, $4$, and $5$ that are initially located on the centroid $O$ of an equilateral triangle, and derive upper bounds on the evacuation time by analyzing their performance. 

\subsubsection{Equal-Travel Strategy}
Consider a straightforward  strategy for evacuating $k\geq 2$ robots, that we call the {\em Equal-travel} strategy: divide the boundary into $k$ contiguous sections and assign each robot to explore a unique section of the boundary. Each robot goes to one endpoint of its assigned section,  it explores it, and then it returns to a meeting point at the centroid to meet  the other robots to share the result of its exploration.

First notice that the time to travel from and to the centroid can vary by almost a factor of 2 for different robots. In particular, robots that are assigned a section of the boundary starting or ending close to a midpoint have to travel a much smaller distance to or from the centroid, than robots that are assigned a section of the boundary that starts/ends close to a vertex. Thus, we divide the boundary into sections in such a way that  any two trajectories of robots are of the same  length and also as short as possible. In other words, we aim to {\em equalize and minimize the  travel time} of all robots. Finally, after meeting at the centroid all robots travel together to the exit that one of them found on the boundary, taking additional time at most $x$. Clearly, at least one of the robots explores at most $3/k$ of the boundary, and needs to travel at most $2x$ to get to and from the boundary;  its trajectory from $O$ and back to $O$, and therefore every robot's trajectory, is thus of length at most $2x+3/k$. 
Thus, this algorithm has evacuation time less than  $3x+3/k=\sqrt{3}+3/k$.  Although very simple, by Theorem \ref{th:krobots}, this bound is asymptotically optimal. As such, we have the following:
\begin{observation}
The Equal-Travel strategy for $k$ robots has worst case evacuation time less than $\sqrt{3}+3/k$.
\end{observation}

In what follows, we propose several improvements of the above strategy by allowing the robots to meet before the whole boundary of the region is explored using an {\em Equal-Travel Early-meeting} strategy for $k\geq 3$, or an {\em Equal-Travel with Detours} strategy for $k=2$.

\subsubsection{Equal-Travel with Detour(s) Algorithms for Two Robots.}
\label{ssec:2robots}
In the Equal-Travel strategy algorithm, when we have two robots, the robots would go together to the center of a side of the triangle and the two robots would meet at the opposite vertex of the triangle at the end of the exploration. However, if the exit is at one of the two vertices explored earlier, this would give an evacuation time of $y+1.5 +1 \approx 2.788$. To eliminate this bad case  we have to adjust the trajectories so that the robot can receive early information about the exit being located in the initial part of the trajectory. This could be done by creating trajectories that meet before the exploration is completed. However, with two robots,
there is no need for the robots to meet. Instead, we can add {\em detour(s)} to the trajectories. As in the  Equal-Travel strategy, trajectories of both robots in the
{\em Equal-Travel with Detour(s)} strategy have the same length. The trajectory for each robot consists of multiple sections of the boundary, with a {\em detour} inside the region between any two sections. Each robot starts with reaching the boundary and exploring a section of the boundary. At some point,  the robot  leaves the boundary and makes a detour inside the triangle (see Figure~\ref{fig:2robotsNewAlg} for an example).
The detour is designed so that if a robot found the exit in its first boundary  section, it can intercept the other robot on the other robot's  corresponding detour. 
If it is not intercepted by the other robot during the detour, the robot concludes that the exit was not found by the other robot in its first section, and goes back to the boundary to  explore its second section of  the boundary. After the robot has finished exploring its second section of the boundary, it can make another detour, and so on.

The idea of using a detour to improve the evacuation time was first proposed in the context of evacuating a disk with two robots in \cite{CzyzowiczGGKMP14}. In this paper, we show that a detour can also improve the evacuation time in an equilateral triangle, and multiple detours can improve it further.

We summarize below the salient features of the detour strategy.
\begin{itemize}
\item The trajectories of the robots are disjoint  except for the initial part to get to the boundary of the triangle, and at the very end. Thus, the detour part of trajectories of robots inside the triangle only get close to each other.
\item Each robot is assigned to explore more than one section of the boundary to explore.
\item A detour is added between every two adjacent sections of the boundary to be explored by a robot.
\item  A robot can do multiple detours.
\item The worst case evacuation time is made the same in each section  of the boundary.
\end{itemize}

We first give the details of an Equal-Travel with Detour algorithm with a single detour. 

\begin{figure}[t]
\centering
\includegraphics[width=0.60\textwidth]{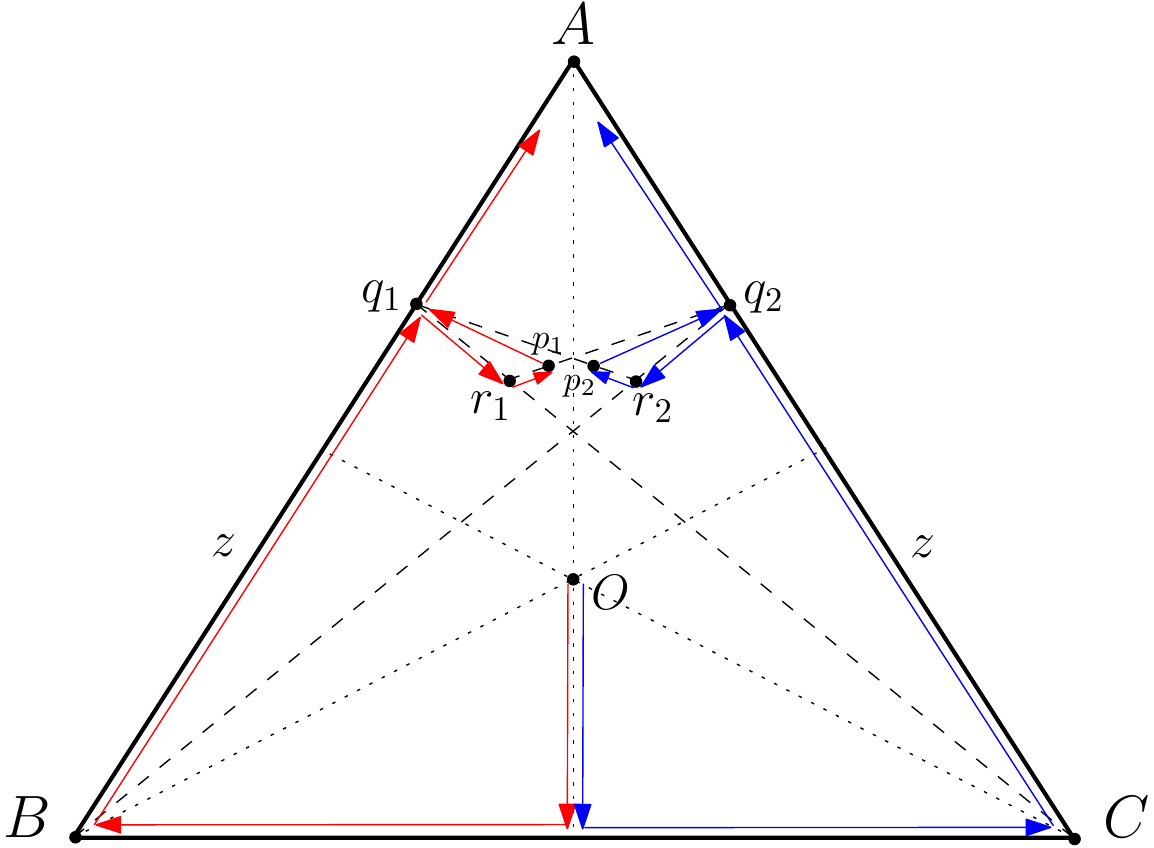}
\caption{The trajectories of two robots with one detour for each robot.}
\label{fig:2robotsNewAlg}
\end{figure}

\begin{theorem}
\label{thm:2robotsNewAlg}
Consider two robots initially located at the centroid of a triangle $T$ with side length 1. There is an Equal-Travel with Detour evacuation algorithm for two robots that  uses a single detour for each robot with the evacuation time $\leq 2.3866$.
\end{theorem}
\begin{proof}
The trajectories of the robots are shown in Figure~\ref{fig:2robotsNewAlg}. Points  $q_1$ and $q_2$ are located symmetrically on the sides $AB$ and $AC$ at distance $z > 0.5 $ (the exact value to be determined later) from $B$ and $C$, respectively. Since the trajectories of robots are symmetric, we specify below the trajectory of $R_2$ only. 
$R_2$ is assigned two sections of the boundary: the section from the midpoint of $BC$ to $C$ and from $C$ to $q_2$, and the section from $q_2$ to $A$. The detour between the two sections consists of segments $q_2r_2$, $r_2p_2$, and $p_2q_2$. Point $r_2$ is located on the segment $q_2B$ and chosen  so that $z +|q_2r_2|=|r_2B|$. Point $p_2$ is located on the segment $p_2q_1$ and chosen so that  $|q_2r_2|+|r_2p_2|=|p_2q_1|$. Thus the positions  of  points on the detour are uniquely determined by $z$.

By the definition of $r_2$ and $p_2$, if  $R_1$ finds the exit at $B$, or  $q_1$, then $R_1$  can intercept $R_2$ at $r_2$, or $p_2$, respectively. Thus if $R_1$ finds the exit on the segment $Bq_1$ then $R_1$  can intercept $R_2$ at a point on the segment $r_2p_2$.

Let $t_1=y+0.5+z+|q_2B|$, and $t_2=y+0.5+z+|q_2r_2|+|r_2p_2|+|p_2q_2| +2(1-z)$. We argue below that the worst case evacuation time of this algorithm is $max(t_1, t_2)$. We consider all possible locations for the exit on the sections of the boundary explored by $R_1$ (the case when the exit is discovered by $R_2$ is symmetric). 

Clearly, if the exit is located on the line segment  $q_1A$, then the evacuation time is at most $t_2$. If the exit is found by $R_1$ on the side $BC$, then it can intercept $R_2$ at point $r_2$ since  $z +|q_2r_2|=|r_2B|$, and the robots reach the exit in time at most $t_1$. Assume now that $R_1$ finds the exit on the line segment $q_1A$ at $D$ as on Figure \ref{figtrboundpr}. Consider the triangle with sides $y_1,y_2,y_3$ formed by segment $BD=y_1$,  side $BA$ and a line though point $D$ parallel with the line going through $q_1$ and $r_2$. Since $z>0.5$, it is easy to see that $|q_1r_2|<|Bq_1|<|Br_2|$ and by the similarity of triangles $y_3<y_1<y_2$. Let $D'$ be the intersection point of the line going through $q_1$ and $r_2$ with the line drawn through $D$ and parallel with $Br_2$.
Since $y_1+|DD'|<y_2+|DD'|=|Br_2|$, when $R_2$ reaches $r_2$ robot $R_1$ is at $D'$ to intercept $R_2$, and the distance traveled by $R_2$ from $r_2$ to $D$
is $y_3+|DD'|<|Br_2|$, and so the total time to reach $D$ is less then $t_1$.
Notice that if the value of $y_1$ is so that the parallel line through $D$ would not intersect the line segment $r_2p_2$ then $R_1$ can intercept $R_2$ at $p_2$ and
the evacuation time remains less than  $t_1$.

We have shown that the worst-case evacuation time is $max(t_1, t_2)$.  We optimize the evacuation time by equating $t_1=t_2$. Solving this  we obtained  $z=0.70745$, which gives  $t_1=2.3866$. Thus the worst-case evacuation time of this algorithm is $2.3866$.
\end{proof}

\begin{figure}
\begin{center}
\includegraphics[width=8cm]{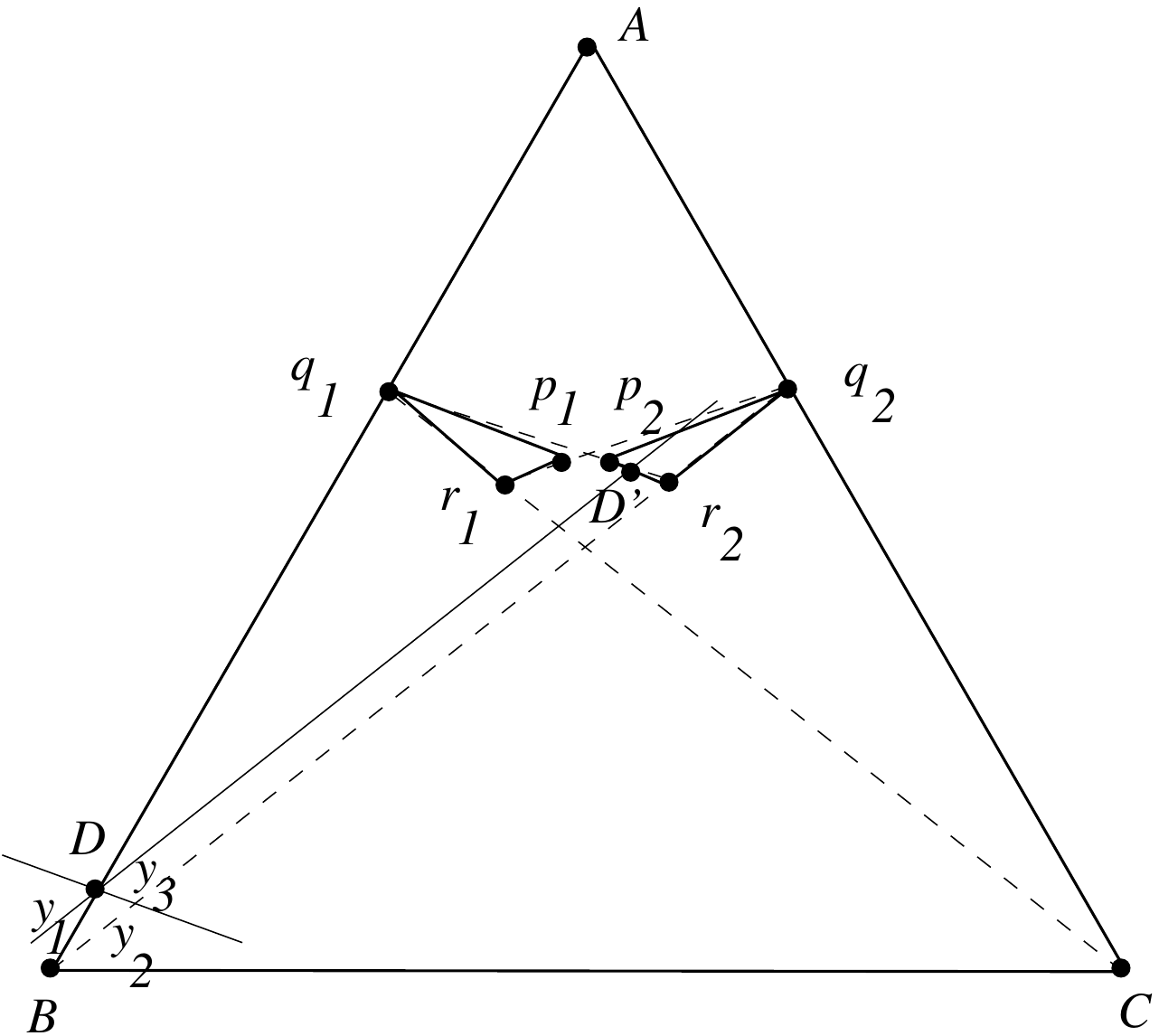}
\end{center}
\caption{Evacuation time for exit at $D$ is less than that for $B$}
\label{figtrboundpr}
\end{figure}

Consider the Equal-Travel with Detour algorithm described above and suppose that the robots do not find the exit in the first phase. That is, the robots will go back to points $q_1$ and $q_2$ to start the exploration of the rest of the boundary of $T$; i.e., the line segments $q_1A$ and $q_2A$. Observe that  at this time in triangle $\triangle Aq_1q_2$, the robots are in the same situation as they were when visiting vertices $B$ and $C$ in the triangle $\triangle ABC$. Furthermore, like for the triangle $\triangle ABC$, where the worst case evacuation time occurs at $B$ and $C$, the worst case for triangle $\triangle Aq_1q_2$ occurs when the exit is located in the neighbourhood of $q_1$ or $q_2$. Therefore,  we can use an additional detour in the triangle $Aq_1q_2$ and improve the evacuation time in the top part as illustrated in Figure~\ref{fig:2robots2LeavesNewAlg}.

Now, each robot is assigned to explore three sections of the boundary. For $R_1$ the first section ends at $q_1$, the second is $q_1q_1'$ and the third $q_1'A$, with one detour added between the sections. Using the same reasoning as in the proof of Theorem \ref{thm:2robotsNewAlg}, the maximum evacuation time in the first, second, third section is when the exit is at $B$, close to $q_1$, close to $q_1'$, respectively. We derive the expressions $t_1$, $t_2$, and $t_3$  for the maximum evacuation time
in each section, and by solving the equations $t_1=t_2, t_2=t_3$ we calculated the optimized positions of $q_1,q_2, q_1',q_2'$ using numerical calculations, obtaining $|Bq_1|=0.666$, $Bq_1'=0.9023$ and the evacuation time $2.3367$. Thus we have the following improved evacuation time.
\begin{theorem}
Consider two robots initially located at the centroid of triangle $T$ with side length 1. There is an Equal-Travel with  Detour evacuation algorithm for two robots that uses two detours for each robot with the evacuation time at most $2.3367$.
\end{theorem}

\begin{figure}[t]
\begin{center}
\includegraphics[width=0.65\textwidth]{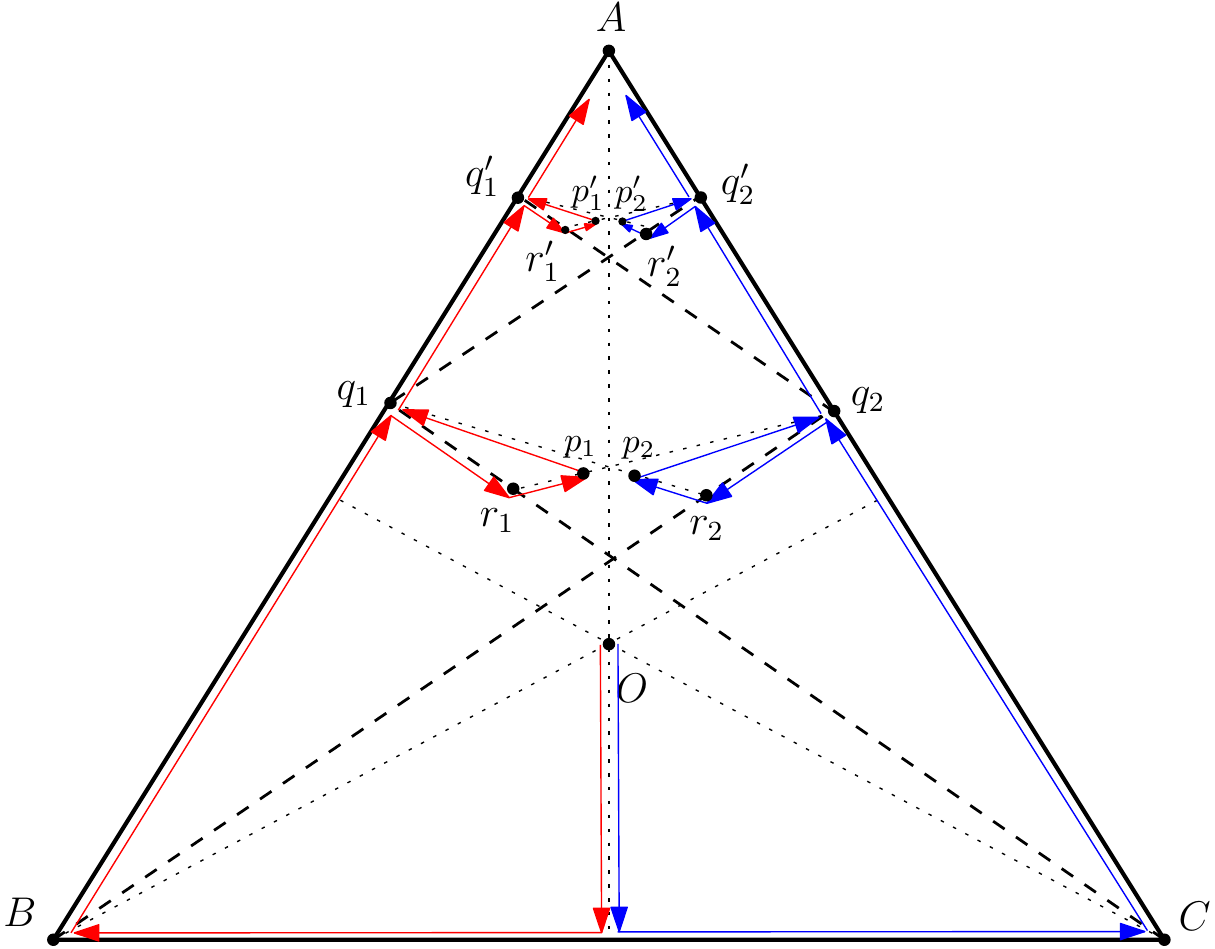}
\end{center}
\caption{Trajectories of two robots with two detours.}
\label{fig:2robots2LeavesNewAlg}
\end{figure}

It is easy to see that the reasoning we used to add a second detour can be applied to the triangle $q_1'Aq_2'$ in the algorithm with two detours per robot to obtain an  Equal-Travel with  Detour evacuation algorithm  with $3$ detours and, recursively,  we can obtain an evacuation algorithm for two robots with $j$ detours for any $j\geq 2$. 
With $j$ detours, the trajectory of a robot consists of $j+1$ segments of the boundary, any two segments separated by a detour. The maximal evacuation time in each segment is when the exit is found at the beginning of a segment, and we have to make the maximal evacuation times equal. This leads to a  system of $j$ equations for the optimal partitioning of the boundary into segments. However, with each additional detour the upper part of the triangle becomes much smaller and the additional segments very short. Already after the second detour, $|q_1'A| <0.1$, and thus the improvement in the evacuation time is  getting extremely tiny for more than two detours. 

\subsubsection{Equal-Travel Early-Meeting Algorithms for $k\geq 3$ Robots}
We derive here better upper bounds for $k\leq 5$, however, our strategies can be easily extended to obtain evacuation algorithms for other values of $k$.

In the simple  Equal-Travel strategy used to derive an upper bound on the evacuation of the triangle by $k$ robots, the robots meet in the interior of the triangle {\em after} the entire boundary has been explored, and then travel together to the exit. For two robots the evacuation time was improved by adding a detour(s) to the trajectories before the entire boundary was explored, as shown in the previous subsection. This allowed the robot that found the exit to inform the other robot about the exit early, making the travel to the exit shorter. The idea of  detours cannot be used with $k\geq 3$, since the robot that found the exit would need to intercept the all other robots on their disjoint detours, leading to a significant delay. In the following we show that for $k\geq 3$ the evacuation time of a triangle can be improved if the robots stop the exploration of the boundary \emph{early} and go to a common meeting point before the boundary is explored entirely. After this \emph{early meeting}, either the robots go together to the exit, or they all go together to explore the rest of the boundary where the exit is now known to be. 

In the sequel, we describe  the salient features of an Equal-Travel Early-Meeting algorithm for $k$ robots that combines the equal travel strategy with an early meeting. 
\begin{itemize}
\item It selects the location of the early meeting point.
\item It divides the boundary into $k+1$ sections, of which $k$ require the same travel times, and are assigned to unique robots, and the last section is a {\em common section}, to be explored together by all robots.
\item In the first phase each robot first explores its assigned section (the last section remains unexplored) and returns to the meeting point. If the exit has been found by one of them, all robots proceed to the exit.
\item If the exit has not been found, it must lie in the unexplored last section. In the second phase the robots go together to explore the common last section and evacuate together.
\end{itemize}

By optimizing  the position of the meeting point, the location of the sections, and the length of the common section, we obtain Early-Meeting algorithms with better performance than those using only the Equal-Travel strategy. Next we give the details of the Early-Meeting algorithm for each $k\in\{3, 4, 5\}$.

\begin{figure}[t]
\begin{center}
\includegraphics[width=6.5cm]{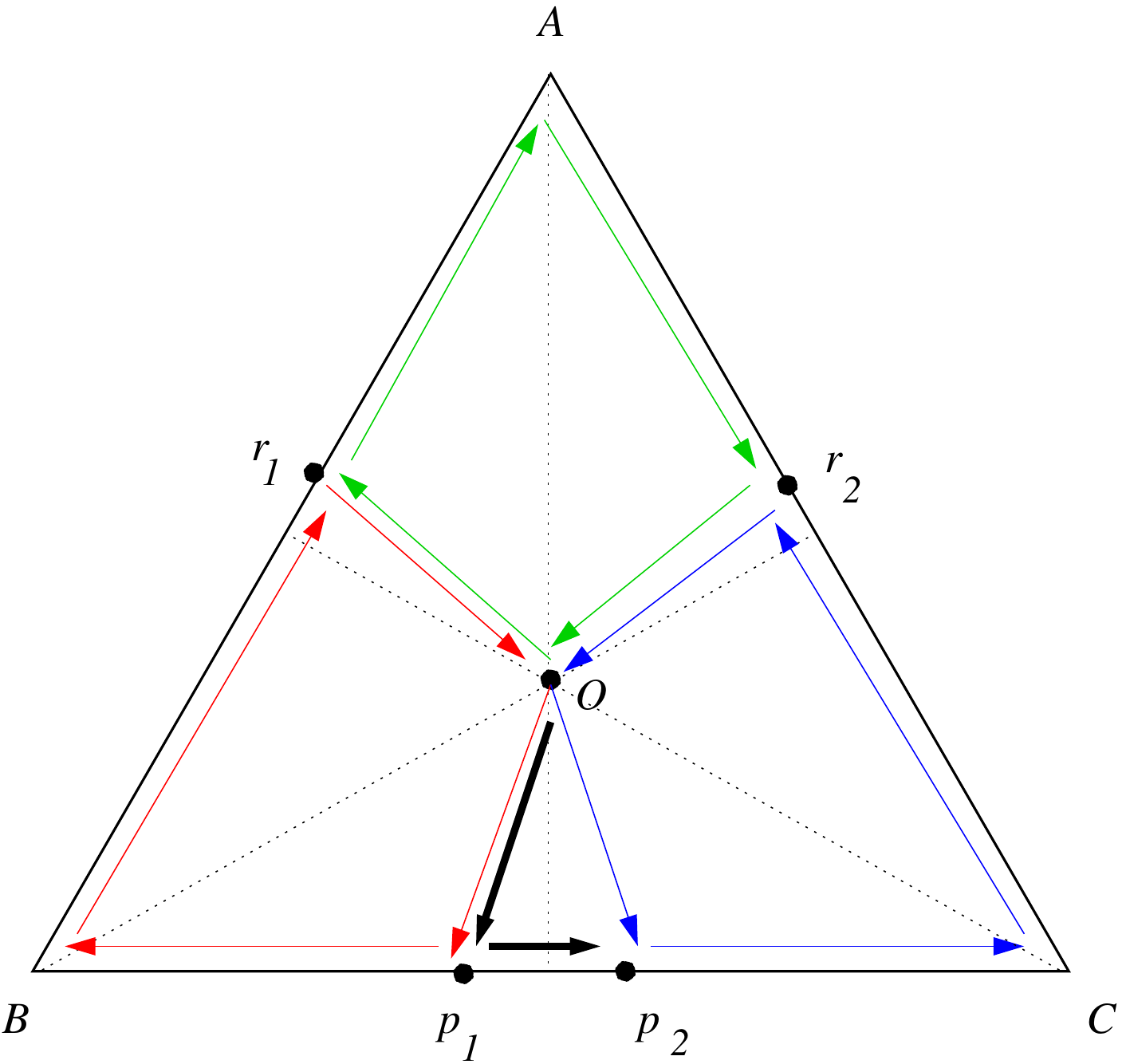}
\end{center}
\caption{Equal-Travel Early-Meeting trajectories for three robots. The common part of trajectories of the robots is shown as thick black line}
\label{fig2and3robots}
\end{figure}

\begin{theorem}
\label{thm:3Robots}
Consider three robots initially located at the centroid of triangle $T$ with side length 1. There are  Equal-Travel Early-Meeting algorithm $\mathcal A$ for three robots with evacuation time $E_{\mathcal A}(T, 3)\leq 2.0887$.
\end{theorem}
\begin{proof}
Consider the following instance of the  Early-Meeting strategy. Place points $p_1, r_1,r_2$ on sides $BC, BA, AC$ as shown in Figure~\ref{fig2and3robots}(b), the exact locations to be determined later.  Centroid $O$ is designated as the meeting point of the robots. Point $p_2$ is placed so that the distance $|Op_1|+|p_1p_2|=x$, the line segment $p_1p_2$ being designated as the common section. $R_1$ is assigned the section of the boundary from $p_1$ to $B$ and to $r_1$, $R_2$  the section from $r_1$ to $A$  to $r_2$, and $R_3$ the section from $r_2$ to $C$  to $p_2$. If one of the robots discovers the exit in its section, the other robots need to travel to it from the meeting point at the centroid, which adds distance at most $x$.
In this setup the maximum distance  robots travel are\\
$t_1=|Op_1|+|p_1B|+|Br_1|+|r_1O|+x$ for $R_1$,\\
$t_2=|Op_2|+|p_2C|+|Cr_2|+|r_2O|+x$ for $R_2$, and\\
$t_3=|Or_1|+|r_1A|+|Ar_2|+|r_2O|+x$ for $R_3$. \\
Solving the set of equations
$t_1=t_2$, $t_2=t_3$ 
and then optimizing the position of $p_1$, we get that the optimal placement of points is $|Bp_1|\approx 0.38601$,
$|Br_1|\approx 0.5252$ and $|Cr_2|\approx 0.5454$, and the evacuation time of the algorithm
at most $2.0888$. All equations were solved and optimizations done using the Maplesoft \cite{th2016}.
\end{proof}

Clearly, the approach used in Theorem~\ref{thm:3Robots} can be generalized to any  $k>3$. Start with selecting a position $p_1$ for the beginning of the common section and partition arbitrarily the boundary by placing points $r_1,r_2,\ldots,r_{k-1}$, see Figure \ref{fig4and5robots} for $k=4$ and 5. Point $p_2$ is placed so that the distance $|Op_1|+|p_1p_2|=x$. Similarly as above, we obtain a system of $k-1$ equations for the values  of $r_1,r_2,\ldots,r_{k-1}$ that produce equal travel time for all robots in the first phase. Finally, by optimizing the value of $p_1$ we get the final assignment of the trajectories of robots. In this manner we obtained the following theorem. 

\begin{figure}
\begin{center}
\includegraphics[width=12.5cm]{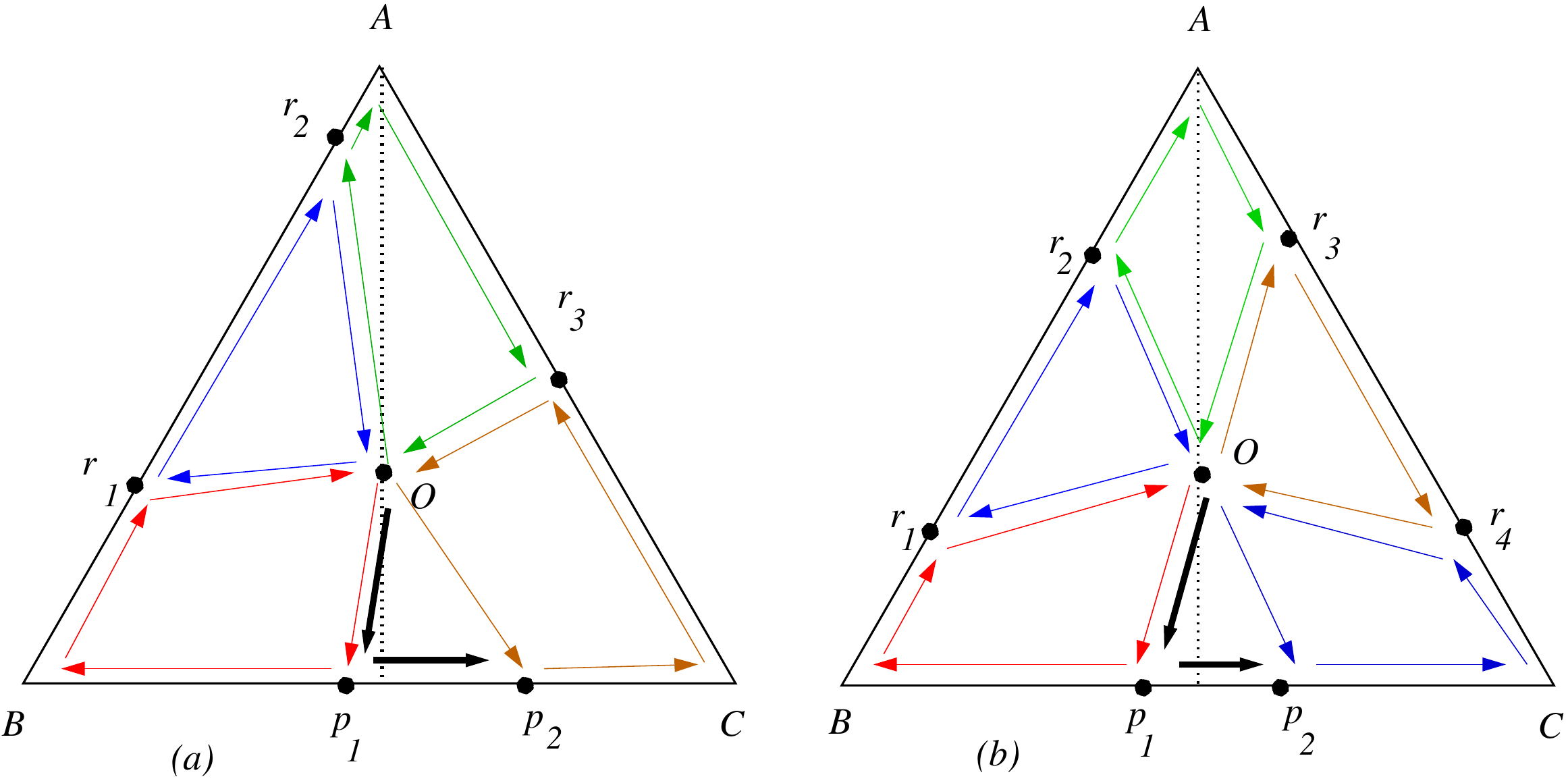}
\end{center}
\caption{Establishing Equal-Travel Early-Meeting Trajectories of (a) four robots, and (b) five robots. The common part of trajectories of the robots is shown as thick black line}
\label{fig4and5robots}
\end{figure}

\begin{theorem}
\label{thm:4Robots}
There are Equal-Travel Early-Meeting algorithms for $k=4$ and  $5$ robots with $E^*(T,4) \leq 1.9816$ and $E^*(T,5) \leq 1.8760$,  respectively.
\end{theorem}
\begin{proof}
The proof is similar to that of Theorem \ref{thm:3Robots} and so it is omitted.
\end{proof}

Unlike for two robots, where we could obtain better evacuation algorithms  by recursively adding more detours, we do not have any evacuation algorithm for three or more robots in which the evacuation time is improved by having more than one meeting point. This is due to the fact that after reaching the meeting point, the remaining exploration  is not similar to the initial configuration and is thus not amenable to recursion.

\section{Evacuation of the Square}
\label{sec:square}
In this section, we consider the problem of evacuating $k=2,3,4$ robots from a square $S$ with side length one.  We denote  the centroid of $S$ by $O$ and denote the vertices of $S$ by $A, B$, $C$ and $D$ as in Figure~\ref{fig:square2robotsLowerBound}. Thus we sometimes write $ABCD$ to refer to $S$. We define $E_{\mathcal{A}}(S,k)$ to be the worst-case evacuation time of the unit-side square $S$ by $k$ robots using algorithm $\mathcal{A}$, and we define  $E^*(S,k)$ to be the {\em optimal} evacuation time of the square by $k$ robots in the face-to-face model. We apply to the square the techniques and strategies that we developed for the equilateral triangle in the preceding section. However important details of proofs and algorithms are different and thus the square needs to be considered separately.

\subsection{Lower Bound}

\begin{figure}[t]
\begin{center}
\includegraphics[width=0.65\textwidth]{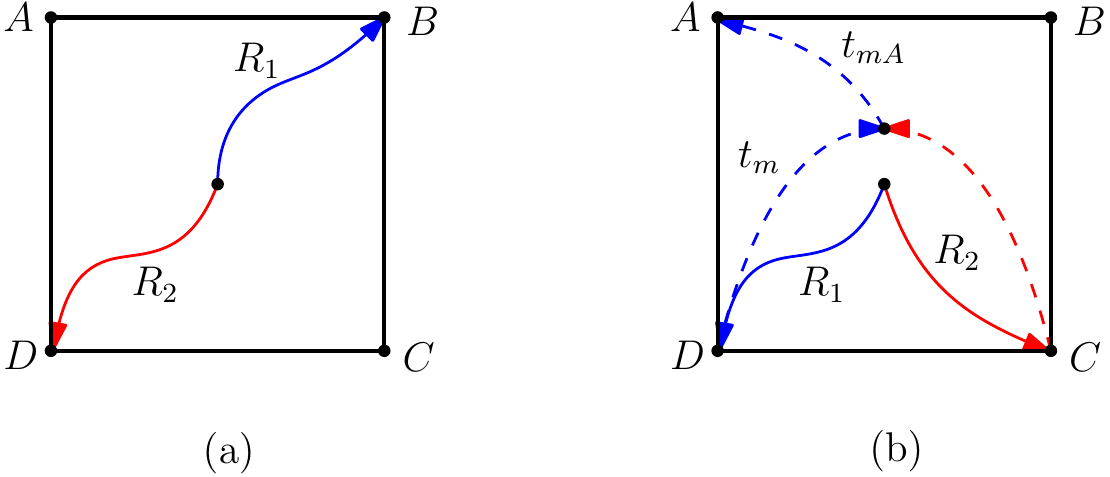}
\end{center}
\caption{An illustration in support of the proof of Theorem~\ref{thm:lb2Square}.}
\label{fig:square2robotsLowerBound}
\end{figure}

\begin{theorem}
\label{thm:lb2Square}
Consider two robots $R_1$ and $R_2$ initially located at the centroid of  square $S$. If the robots communicate using the face-to-face model, then the evacuation time of these two robots is $E^*(S,2)\geq 1+3\sqrt{2}/2\approx 3.121$.
\end{theorem}
\begin{proof}
Suppose for a contradiction that there exists an algorithm $\mathcal{A}$ that evacuates two robots $R_1$ and $R_2$ from $S$ before time $1+3\sqrt{2}/2$. Let $v_1, v_2, v_3, v_4$ be the order in which the four vertices of $S$ are visited by the algorithm, at times $t_1, t_2, t_3, t_4$ respectively.  If $t_3 \geq 3 \sqrt{2}/2$, then if the exit is at $v_4$, since $dist(v_3, v_4) \geq 1$, we obtain an evacuation time of at least $1+3\sqrt{2}/2$, a contradiction. So we conclude that $t_3 < 3\sqrt{2}/2$. Since a single robot requires time at least $2 + \sqrt{2}/2 > 3 \sqrt{2}/2$ to visit 3 vertices, and time at least $3 \sqrt{2}/2$ to visit two diagonally opposite vertices, it must be that one robot visited 2 {\em adjacent} vertices and the other robot the remaining one of $\{v_1, v_2, v_3 \}$ at or before  time $t_3$. Without loss of generality, we assume that $R_1$ visited $D$ and then $A$. Let $t_v$ denote the time of first visit for vertex $v$. Notice that $1 + \sqrt{2}/2 \leq t_A < 3 \sqrt{2}/2$.  If $v_4$ is $B$, and the exit is at $B$ then clearly $R_1$ needs time $1 + 3 \sqrt{2}/2$ to arrive at $B$, a contradiction. Therefore it must be that $R_2$ visited $B$ at time $\leq t_3$ and  $v_4=C$. We now claim (as in the Meeting Lemma for the triangle) that the trajectory for $R_1$ cannot be such that $R_1$ meets $R_2$ after $R_2$ visits $B$ and before $t_A$, because this would imply that  $t_A \geq t_B + |BM| + |MA| \geq \sqrt{2}/2 + \sqrt{2} = 3 \sqrt{2}/2$, a contradiction. Therefore, if the exit is at $B$, $R_1$ cannot reach $B$ before time $t_A + \sqrt{2} \geq 1 + 3 \sqrt{2}/2$.
\end{proof}

\subsection{Evacuation Algorithms}
In this subsection  we derive upper bounds on the evacuation time by giving the evacuation algorithms for specific values of $k$. Similarly as for the triangle, we give an  Equal-Travel with Detour algorithm for $k=2$, and  Equal-Travel with Early Meeting algorithms for $k \geq 3$.

\subsubsection{Equal-Travel with Detour for Two Robots}
For $k=2$ the simple equal travel algorithm, which goes initially to the center of side $DC$, has the maximal evacuation time of $2.5$,  achieved when  the exit is located at $C$ or $D$.  
Consider the  Equal-Travel with Detour algorithm evacuation algorithm $\mathcal{A}$ shown in Figure~\ref{fig:square2robotsAlg}. We first explain the trajectories of each robot and then will discuss the details of each section of the trajectories, depending on where the exit is located.

Starting from the centroid $O$, the robots  go together to $F$. Then, $R_1$ explores the boundary from $F$ to  $D$, to $A$, and then to $E$. There it takes the detour in the direction of $C$ to $G$, from $G$ in the direction of $J$ until $I$, and then returns to $E$.  From $E$ it explores the second section of the boundary to $M$. The  positions of $E$, $G$, $I$, and $J$ are specified below. Robot $R_2$ follows the symmetric trajectory that explores the other part of the boundary with one symmetric detour. 

\begin{figure}[t]
\begin{center}
\includegraphics[width=0.45\textwidth]{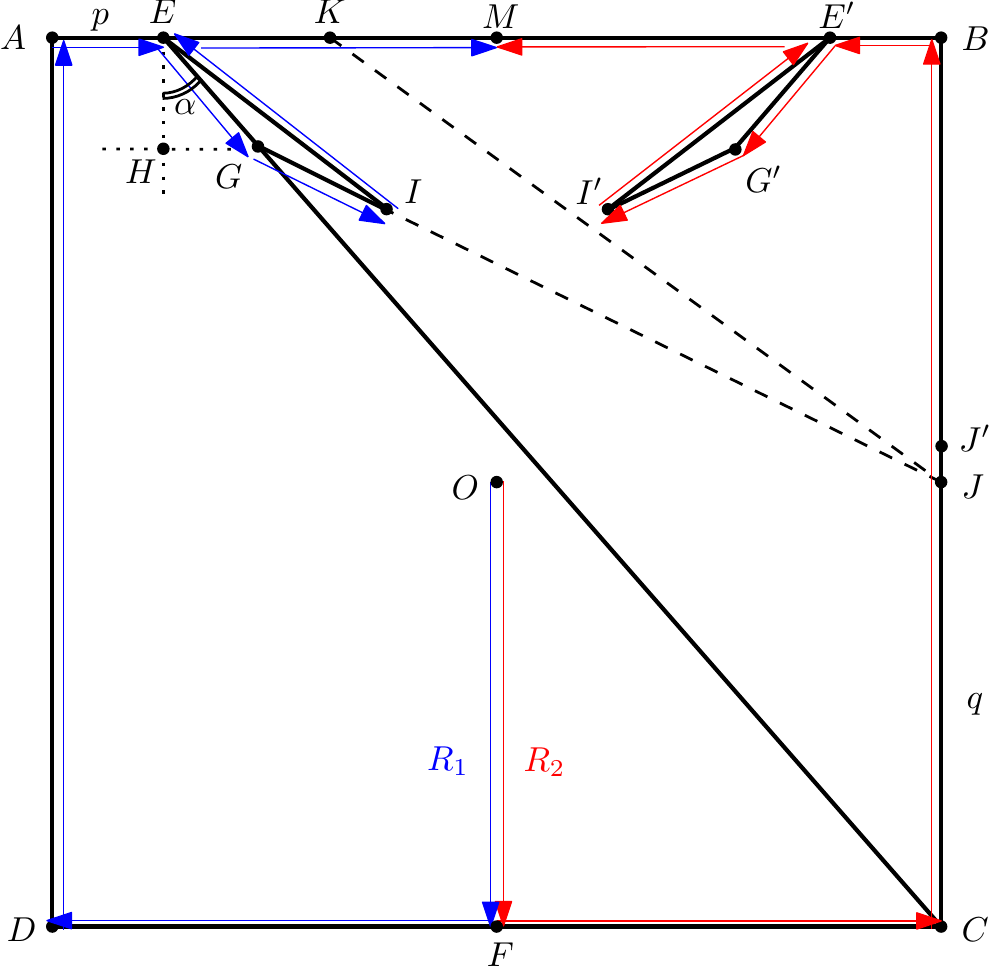}
\end{center}
\caption{Trajectories of two robots for evacuating a square.}
\label{fig:square2robotsAlg}
\end{figure}

\begin{enumerate}
\item The positions of $E$ and $J$ are parameters whose values are optimized by numerical calculations, $p<0.25$ and $q<1-|EH|$.
\item The  point $G$ is chosen so that $|DA|+|AE|+|EG|=|CG|$. This ensures that if the exit is found on the segment $FC$, then $R_2$ intercepts $R_1$ at $G$ and informs it about the exit.
\item The point $I$ is chosen such that $|CJ|+|JI|=|DA|+|AE|+|EG|+|GI|$. This ensures that if the exit is somewhere on the segment $CJ$, then $R_2$ intercepts $R_1$ at point $I$ and informs it about the exit.
\end{enumerate}

Now, consider the case when the exit is at a point $J'$ located on segment $JB$ arbitrarily close to  point $J$ (i.e., $|JJ'|=\epsilon$ for some small $\epsilon>0$). Then $R_2$ can intercept $R_1$ only after $R_1$ finishes its detour, since $|JI|<|JJ'|+|J'I|$. That is, $R_2$ can intercept $R_1$ at some point $K$ on segment $EM$ and inform it about the exit. Clearly, $K$ must be to the left of $M$, since even without taking a shortcut the trajectories of the two robots are of the same length and they meet at $M$. Notice that this implies  a discontinuity in the evacuation time in  at $J$ from above $J$.

If the exit is at a point on segment $JB$ then $R_2$ can intercept $R_1$ on the segment $KM$. If the exit is somewhere on $BM$, then the two robots meet at $M$ and then go to the exit together from there.

We first show that due to our selection of the trajectories, the maximum evacuation time occurs at one of the three points $C, J, B$, or is  equal to $\lim_{e\leftarrow 0}EvacuationTime(J')$. By constraint (1), the maximum evacuation time  when the exit is located on $FC$ occurs when it is at $C$. In Lemma~\ref{lem:shorter} below, we show that the maximum evacuation time  when the exit is on $CJ$ occurs when it is at $C$.  However, as mentioned above, the evacuation time is greater than that at $J$ if the exit is just above $J$, say at $J'$ (i.e., at a point that is right above $J$); this is because $R_2$ will then need to intercept $R_1$ at a point $K$ on $EB$. After this the evacuation time decreases as $R_2$ finds the exit further away  from $C$ up to $J$. Finally, it is easy to see that after this jump  the evacuation time decreases as $R_2$ approaches  $B$, and also from $B$ to $M$.

\begin{figure}[t]
\begin{center}
\includegraphics[width=0.45\textwidth]{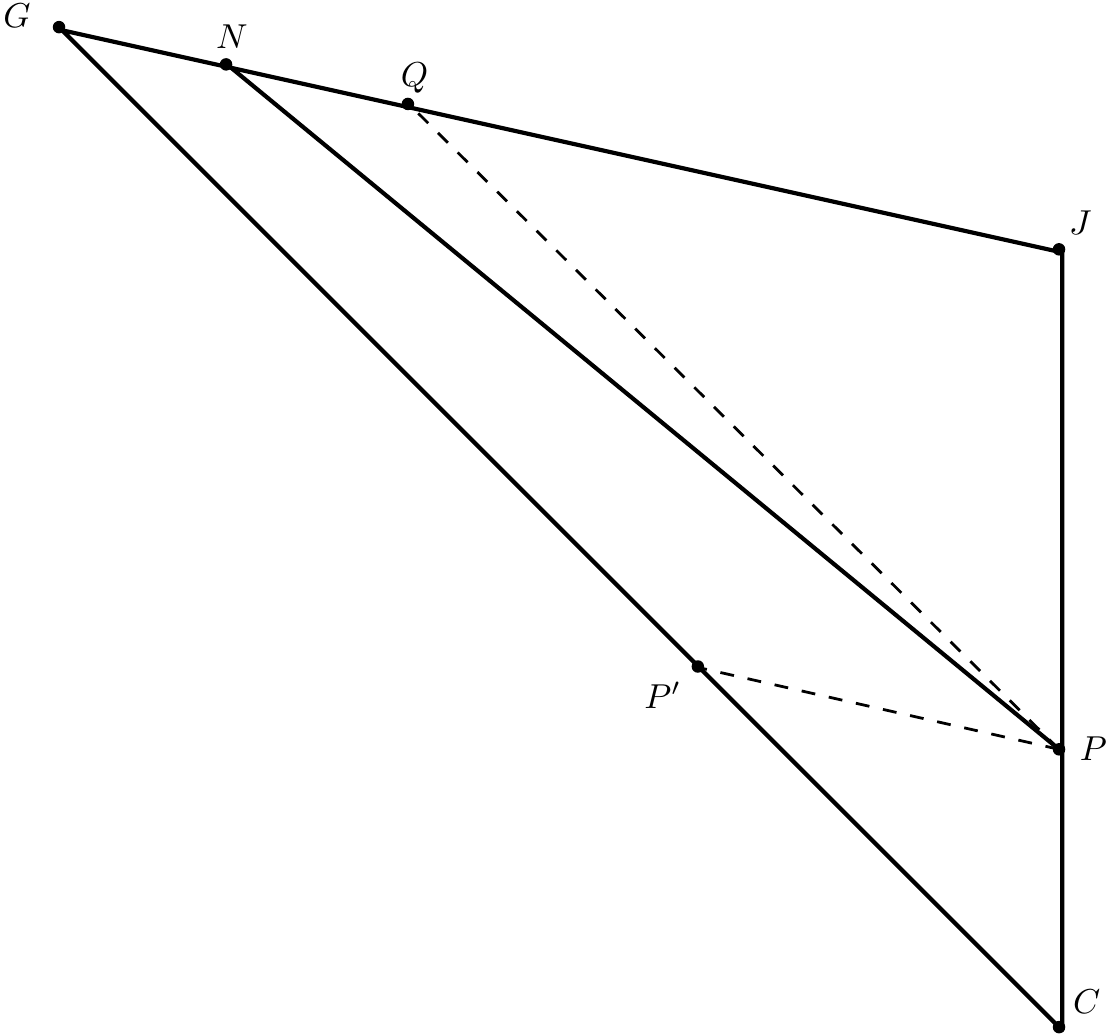}
\end{center}
\caption{An illustration  supporting the proof of Lemma~\ref{lem:shorter}.}
\label{fig:shorter}
\end{figure}

\begin{lemma}
\label{lem:shorter} 
Assume that $R_2$ finds the exit at $P$ on segment $CJ$ where $P \neq C$.  Then the evacuation time  at $P$ is less than the evacuation time at $C$.
\end{lemma}
\begin{proof}
If the exit is found at $P$ then $R_2$ intercepts $R_1$ at point $N$ on the segment $GI$ as in Figure \ref{fig:shorter}. Consider the triangle $GJC$. Obviously, $\angle GJC>\pi/2$. Let $Q$ be a  point on  $GJ$ such that $PQ$ and $CG$ are parallel.  Consider the line passing through $P$ and parallel to $GJ$ and let $P'$ be the intersection point of this line and $GC$; see Figure~\ref{fig:shorter}. Since $\angle P'PC>\pi/2$, the edge $P'C$ has the largest length among all the three edges of triangle $PP'C$; hence, $|P'C|>|PP'|$. Therefore,
\[
|GC|=|GP'|+|P'C|=|QP|+|P'C|>|QP|+|PP'|=|QP|+|GI|=|QP|+|GN|+|NQ|.
\]
By the triangle inequality for triangle $NQP$, we have $|NQ|+|QP|>|NP|$. Therefore, 
$$|GC|>|GN|+|NP|.$$
To complete the proof, observe that the evacuation time at $C$ is $2+|AE| + |EG| + |GC|$ while the evacuation time at $P$ is $2+|AE| + |EG| + |GN| + |NP|$, proving the result.
\end{proof}

We now compute the evacuation time for each of these critical points. Let $p=|AE|$ and $q=|CJ|$. Then, the angle $\angle HEG$ as shown in Figure~\ref{fig:square2robotsAlg} is
\[
\alpha=\arcsin(\frac{1-p}{\sqrt{1+(1-p)^2}}).
\]
By having $\alpha$, we can compute $|EG|, |GI|$ and $|EI|$, which gives us the time\\ $detour=|EG|+|GI|+|EI|$ to traverse the detour.
Then, we have  the following evacuation times:\\
 \hspace*{4mm}$1+2\;|GC|$ if the exit is at $C$,\\
  \hspace*{4mm}$1+q+2\; |IJ|$ if the exit is at $J$,\\
  \hspace*{4mm}$1+q+2\; |KJ|$ if the exit is at $J'$, and \\
  \hspace*{4mm}$3+detour$ if the exit is at $B$,\\
where $|GC|, |IJ|$ and $|KJ|$ can be expressed only in terms of $p$ and $q$. Using Python code, we computed each of the above four evacuation times for $0.1\leq p\leq 0.2$ with step size $0.00005$ and $0.2\leq q\leq 0.8$ with step size $0.00005$ to find the values of $p$ and $q$ that optimize the evacuation times. We obtained that for $p\approx 0.1556$ and $q\approx 0.5010$ the evacuation time is $E_{\mathcal{A}}(S,2)\leq 3.4644$. Thus we have the following theorem.

\begin{theorem}
\label{thm:square2robotsAlg}
Consider two robots initially located at the centroid of square $S$ with side length 1. There is an Equal-Travel with Detour algorithm $\mathcal{A}$ for two robots with evacuation time $E_{\mathcal{A}}(S,2)\leq 3.4644$ in the face-to-face model.
\end{theorem}

In our algorithm  $\mathcal{A}$, the evacuation time of  $3.4644$ is reached at $C$ and also $J'$ and is strictly lower everywhere else. Notice that the evacuation time at $J'$ can be decreased if $R_1$ makes a ``smaller'' second detour on $EK$ just prior to reaching $K$, so that robot $R_2$ can intercept $R_1$ on the second detour. By decreasing the evacuation time at $J'$ we can then improve the evacuation times at $C$, which slightly decreases the overall evacuation time. Thus, additional detour(s) can be used recursively in the evacuation algorithms in the square for the $k=2$, similarly as shown for the triangle.     

\begin{figure}[t]
\begin{center}
\includegraphics[width=0.80\textwidth]{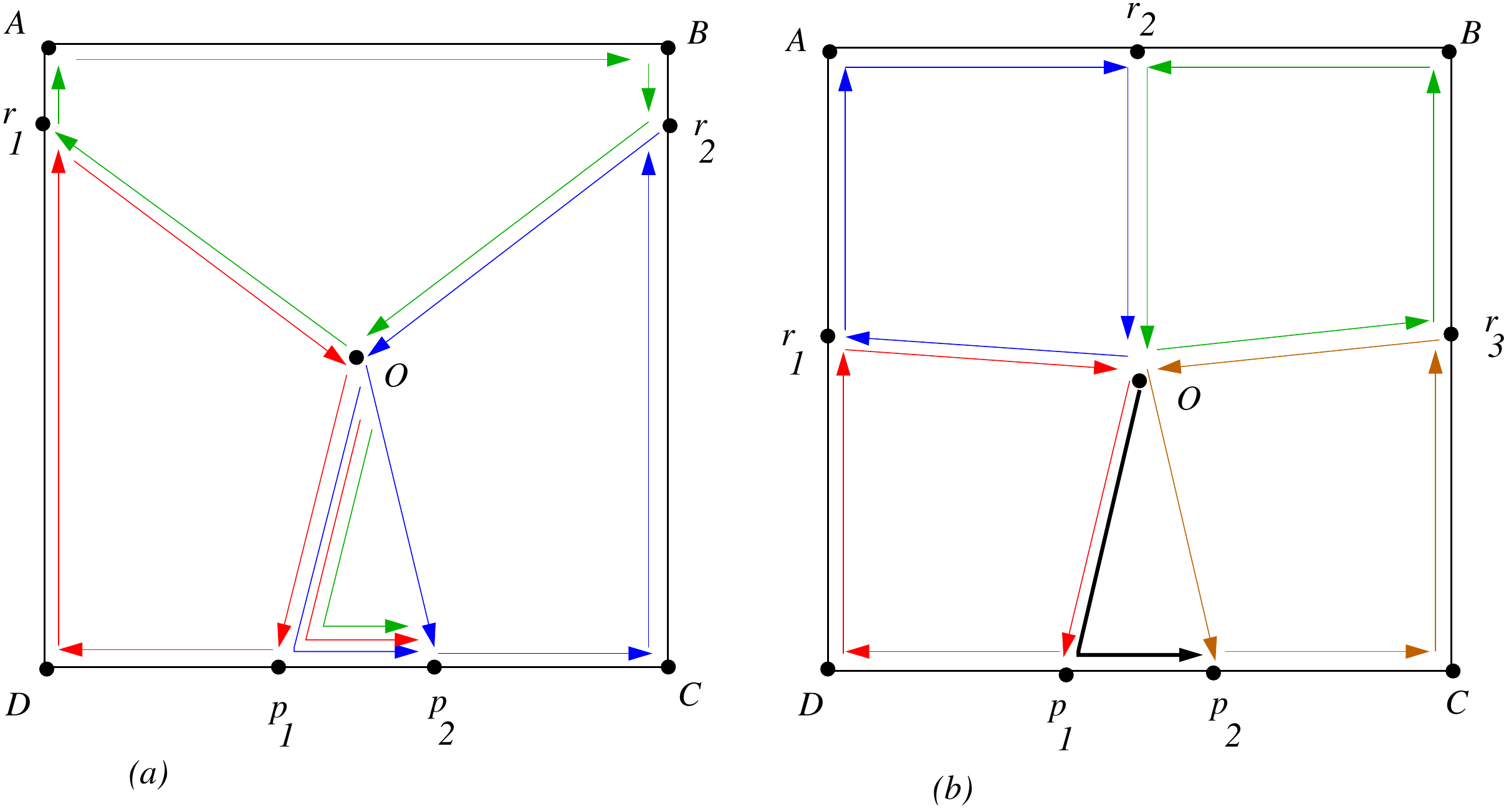}
\end{center}
\caption{Trajectories of three and four robots for evacuating a square. Section $p_1p_2$ is a common part of trajectories  of all robots.}
\label{fig:square3and4robots}
\end{figure}

\subsubsection{Equal-Travel Early-Meeting Algorithms for $k\geq 3$ Robots}
We now consider the case of more than two robots. We specify trajectories for the  Equal-Travel Early-Meeting strategy, and determine the evacuation time of these algorithms for $k=3$ and $k=4$.
\begin{theorem}
\label{thm:square3and4}
There are Equal-Travel Early-Meeting algorithms for three and four robots initially located in the centroid of  square $S$ with side length 1 with evacuation times $E^*(S,3)\leq 3.178$ and $E^*(S,4)\leq 2.664$, respectively.
\end{theorem}
\begin{proof}
The algorithms for both $k=3$ and $k=4$ are similar to the Equal-Travel Early-Meeting algorithm for $k=3$ in case of the equilateral triangle (see Theorem~\ref{thm:3Robots}). The centroid $O$ is used as the meeting point. Place the points $p_1$ and $p_2$ on the side $DC$ as shown in Figure~\ref{fig:square3and4robots}(a); the exact locations will be determined later. Let $|Op_1|+|p_1p_2|=\sqrt{2}/2$. The remainder of the boundary is divided using points $r_1,r_2,\ldots, r_k$  into $k$ segments.

When $k=3$, robot $R_1$ is assigned the section of the boundary from $p_1$ to $D$, to $r_1$, and then it goes back to $O$, $R_2$ is assigned the section from $r_1$ to $A$, $A$ to $B$, $B$ to $r_2$ and then $r_2$ back to $O$, and $R_3$ is assigned the section form $r_2$ to $C$, $C$ to $p_2$ and then $p_2$ back to $O$. If the exit is discovered by one of the robots in its section, then the other robots need to travel to it from the meeting point at the centroid, which adds distance at most $\sqrt{2}/2$. Otherwise the robots travel together to $p_1$ to explore $p_1p_2$, also for time  $\sqrt{2}/2$.  Therefore, the maximum distance traveled by robots is one of\\
$t_1=|Op_1|+|p_1D|+|Dr_1|+|r_1O|+\sqrt{2}/2$ for $R_1$,\\
$t_2=|Or_1|+|r_1A|+|AB|+|Br_2|+|r_2O|+\sqrt{2}/2$ for $R_2$, and\\
$t_3=|Or_2|+|r_2C|+|Cp_2|+|p_2O|+\sqrt{2}/2$ for $R_3$.

Solving the set of equations $t_1=t_2$, $t_2=t_3$ and then optimizing the position of $p_2$ as we did in the proof of Theorem~\ref{thm:3Robots}, we get $|Dp_1|=|Cp_2|=0.4012$ and $|Dr_1|=|Cr_2|=0.9124$. This gives us $E^*(S,3)\leq 3.178$.

When $k=4$, the trajectories of the robots are shown in Figure~\ref{fig:square3and4robots}(b). By an analogous calculations  those above for $k=3$, we get $|Dp_1|=|Cp_2|=0.4012$, $|Cr_3|=|Dr_1|=0.5445$ and $|Ar_2|=0.5$. This gives us $E^*(S,4)\leq 2.664$.
\end{proof}

Clearly, our algorithms for $k=3$ and $4$ can be easily generalized for any $k>4$.

\section{Discussion}
\label{sec:conclusion}
We studied the evacuation of an equilateral triangle or a square by $k$ robots, initially located at its centroid,  when the robots can communicate with each other only when they meet; i.e., they use face-to-face communication.

For the triangle we showed a lower bound of $\sqrt{3}$ on the evacuation time by any number $k$ of robots, and gave a simple Equal-travel strategy that achieves this bound asymptotically. For $k=2$ robots,  we proved a lower bound of $1 + 2 /\sqrt{3} \approx 2.154$.  We then showed that for $k=2$  the Equal-travel strategy can be improved   by  adding  detours  in the interior of the triangle,  and obtained an upper bound of $2.3367$ on the evacuation time. This upper bound is achieved with two detours by each robot. We showed that detours can be used recursively to improve the evacuation time. For $k \geq 3$, we studied the  Equal-Travel Early-Meeting strategy for evacuation algorithms in which the robots meet at an {\em early  meeting point} inside the triangle before the whole boundary is examined. This strategy  gave us upper bounds of  $2.08872$, $1.9816$, and
$1.876$ for $k=3, 4$ and 5, respectively. We then showed that the same strategies can be applied to obtain evacuation algorithms for a unit square.

Our work points to  a number of directions open for future work. First, closing the existing gaps in our results for the evacuation time in the face-to-face model remains open for both the equilateral triangle and the square. Moreover, although we limited our study to the evacuation of the equilateral triangle and square of unit side, the lower bound and algorithmic strategies used in this paper should be applicable to other convex search domains. Finally, a clear understanding of the search domains in which early meetings or detours are {\em provably} useful remains elusive.

\subparagraph{Acknowledgment.} The authors thank Iman Bagheri for catching and helping to fix an error in the proof of Lemma 6. The second author thanks Prosenjit Bose for useful discussions on the problem.

\bibliographystyle{plain}
\bibliography{ref}

\newpage
\appendix
\section*{Appendix A}
\label{sec:apxA}
The Python source code for computing the evacuation time claimed in Theorem~\ref{thm:square2robotsAlg} is given below.
\begin{verbatim}
import math

p=0.1
q=0.2
bestP=p
bestQ=q
minimum=5
while p<0.2:
    q=0.2
    while q<0.8:
        alpha=math.asin((1-p)/math.sqrt(1+math.pow(1-p,2)))
        ec=math.sqrt(1+math.pow(1-p,2))
        x=(ec-1-p)/2
        eg=x
        hg=x*math.sin(alpha)
        gc=ec-eg
        jc=q
        gj=math.sqrt(gc*gc+jc*jc-2*gc*jc*math.cos(alpha))
        #We must have 1+p+x+gi = q+ij
        f=1+p+x
        
        ij=(f-q+gj)/2
        gi=gj-ij

        #Time to meet:
        t1=1+q+ij
        t2=1+1+p+x+gi

        #The remaining time:
        #We need ic
        cosBeta=(gc*gc-gj*gj-jc*jc)/(-2*gj*jc)
        beta=math.acos(cosBeta)
        ic=math.sqrt(ij*ij+jc*jc-2*ij*jc*cosBeta)

        ei=math.sqrt(eg*eg+gi*gi-2*eg*gi*math.cos(alpha+beta))

        detour=eg+gi+ei

        lhs=1+detour-q
        A=1+detour-q

        r=(2+q*q-2*q-A*A)/(2*A+2)
        kj=math.sqrt((1-r)*(1-r)+(1-q)*(1-q))

        m1=max(1+2*gc, 1+q+2*ij)
        m2=max(1+q+2*kj, 3+detour)
        maximum=max(m1, m2)

        if maximum<minimum:
            minimum=maximum
            bestP=p
            bestQ=q

        q+=0.00005
    p+=0.00005

print('The evacuation time is:', minimum)
print('The value of p is:', bestP)
print('The value of q is:', bestQ)
\end{verbatim}

The output of the above source code is shown below.
\begin{verbatim}
The evacuation time is: 3.4644205988599026
The value of p is: 0.15559999999999388
The value of q is: 0.5009999999999669
\end{verbatim}

\end{document}